\documentclass[pra,aps,showpacs,twocolumn,twoside,superscriptaddress]{revtex4}

\usepackage{amsmath,amsfonts,amssymb,caption,color,epsfig,graphics,graphicx,latexsym,mathrsfs,revsymb,theorem,url,verbatim,epstopdf}
\usepackage{float}
\usepackage{hyperref}

\hypersetup{colorlinks,linkcolor={blue},citecolor={blue},urlcolor={red}}

\usepackage{hyperref}

\newtheorem{definition}{Definition}
\newtheorem{proposition}[definition]{Proposition}
\newtheorem{lemma}[definition]{Lemma}

\newtheorem{theorem}[definition]{Theorem}
\newtheorem{corollary}[definition]{Corollary}
\newtheorem{conjecture}[definition]{Conjecture}

\newtheorem{remark}[definition]{Remark}
\newtheorem{example}[definition]{Example}
\newtheorem{question}[definition]{Question}

\def\squareforqed{\hbox{\rlap{$\sqcap$}$\sqcup$}}
\def\qed{\ifmmode\squareforqed\else{\unskip\nobreak\hfil
		\penalty50\hskip1em\null\nobreak\hfil\squareforqed
		\parfillskip=0pt\finalhyphendemerits=0\endgraf}\fi}
\def\endenv{\ifmmode\;\else{\unskip\nobreak\hfil
		\penalty50\hskip1em\null\nobreak\hfil\;
		\parfillskip=0pt\finalhyphendemerits=0\endgraf}\fi}
\newenvironment{proof}{\noindent \textbf{{Proof.~} }}{\qed}
\def\Dbar{\leavevmode\lower.6ex\hbox to 0pt
	{\hskip-.23ex\accent"16\hss}D}
\makeatletter
\def\url@leostyle{%
	\@ifundefined{selectfont}{\def\UrlFont{\sf}}{\def\UrlFont{\small\ttfamily}}}
\makeatother
\def\bcj{\begin{conjecture}}
	\def\ecj{\end{conjecture}}
\def\bcr{\begin{corollary}}
	\def\ecr{\end{corollary}}
\def\bd{\begin{definition}}
	\def\ed{\end{definition}}
\def\bea{\begin{eqnarray}}
\def\eea{\end{eqnarray}}
\def\bem{\begin{enumerate}}
	\def\eem{\end{enumerate}}
\def\bex{\begin{example}}
	\def\eex{\end{example}}
\def\bim{\begin{itemize}}
	\def\eim{\end{itemize}}
\def\bl{\begin{lemma}}
	\def\el{\end{lemma}}
\def\bma{\begin{bmatrix}}
	\def\ema{\end{bmatrix}}
\def\bpf{\begin{proof}}
	\def\epf{\end{proof}}
\def\bpp{\begin{proposition}}
	\def\epp{\end{proposition}}
\def\bqu{\begin{question}}
	\def\equ{\end{question}}
\def\br{\begin{remark}}
	\def\er{\end{remark}}
\def\bt{\begin{theorem}}
	\def\et{\end{theorem}}

\def\btb{\begin{tabular}}
	\def\etb{\end{tabular}}

\newcommand{\nc}{\newcommand}

\def\a{\alpha}
\def\b{\beta}
\def\g{\gamma}
\def\d{\delta}
\def\e{\epsilon}

\def\t{\theta}

\def\p{\pi}
\def\r{\rho}

\def\ps{\psi}
\def\o{\omega}

\nc{\bbA}{\mathbb{A}} \nc{\bbB}{\mathbb{B}} \nc{\bbC}{\mathbb{C}}
\nc{\bbD}{\mathbb{D}} \nc{\bbE}{\mathbb{E}} \nc{\bbF}{\mathbb{F}}
\nc{\bbG}{\mathbb{G}} \nc{\bbH}{\mathbb{H}} \nc{\bbI}{\mathbb{I}}
\nc{\bbJ}{\mathbb{J}} \nc{\bbK}{\mathbb{K}} \nc{\bbL}{\mathbb{L}}
\nc{\bbM}{\mathbb{M}} \nc{\bbN}{\mathbb{N}} \nc{\bbO}{\mathbb{O}}
\nc{\bbP}{\mathbb{P}} \nc{\bbQ}{\mathbb{Q}} \nc{\bbR}{\mathbb{R}}
\nc{\bbS}{\mathbb{S}} \nc{\bbT}{\mathbb{T}} \nc{\bbU}{\mathbb{U}}
\nc{\bbV}{\mathbb{V}} \nc{\bbW}{\mathbb{W}} \nc{\bbX}{\mathbb{X}}
\nc{\bbZ}{\mathbb{Z}}

\nc{\bA}{{\bf A}} \nc{\bB}{{\bf B}} \nc{\bC}{{\bf C}}
\nc{\bD}{{\bf D}} \nc{\bE}{{\bf E}} \nc{\bF}{{\bf F}}
\nc{\bG}{{\bf G}} \nc{\bH}{{\bf H}} \nc{\bI}{{\bf I}}
\nc{\bJ}{{\bf J}} \nc{\bK}{{\bf K}} \nc{\bL}{{\bf L}}
\nc{\bM}{{\bf M}} \nc{\bN}{{\bf N}} \nc{\bO}{{\bf O}}
\nc{\bP}{{\bf P}} \nc{\bQ}{{\bf Q}} \nc{\bR}{{\bf R}}
\nc{\bS}{{\bf S}} \nc{\bT}{{\bf T}} \nc{\bU}{{\bf U}}
\nc{\bV}{{\bf V}} \nc{\bW}{{\bf W}} \nc{\bX}{{\bf X}}
\nc{\bZ}{{\bf Z}}

\nc{\cA}{{\cal A}} \nc{\cB}{{\cal B}} \nc{\cC}{{\cal C}}
\nc{\cD}{{\cal D}} \nc{\cE}{{\cal E}} \nc{\cF}{{\cal F}}
\nc{\cG}{{\cal G}} \nc{\cH}{{\cal H}} \nc{\cI}{{\cal I}}
\nc{\cJ}{{\cal J}} \nc{\cK}{{\cal K}} \nc{\cL}{{\cal L}}
\nc{\cM}{{\cal M}} \nc{\cN}{{\cal N}} \nc{\cO}{{\cal O}}
\nc{\cP}{{\cal P}} \nc{\cQ}{{\cal Q}} \nc{\cR}{{\cal R}}
\nc{\cS}{{\cal S}} \nc{\cT}{{\cal T}} \nc{\cU}{{\cal U}}
\nc{\cV}{{\cal V}} \nc{\cW}{{\cal W}} \nc{\cX}{{\cal X}}
\nc{\cZ}{{\cal Z}}

\nc{\hA}{{\hat{A}}} \nc{\hB}{{\hat{B}}} \nc{\hC}{{\hat{C}}}
\nc{\hD}{{\hat{D}}} \nc{\hE}{{\hat{E}}} \nc{\hF}{{\hat{F}}}
\nc{\hG}{{\hat{G}}} \nc{\hH}{{\hat{H}}} \nc{\hI}{{\hat{I}}}
\nc{\hJ}{{\hat{J}}} \nc{\hK}{{\hat{K}}} \nc{\hL}{{\hat{L}}}
\nc{\hM}{{\hat{M}}} \nc{\hN}{{\hat{N}}} \nc{\hO}{{\hat{O}}}
\nc{\hP}{{\hat{P}}} \nc{\hR}{{\hat{R}}} \nc{\hS}{{\hat{S}}}
\nc{\hT}{{\hat{T}}} \nc{\hU}{{\hat{U}}} \nc{\hV}{{\hat{V}}}
\nc{\hW}{{\hat{W}}} \nc{\hX}{{\hat{X}}} \nc{\hZ}{{\hat{Z}}}

\nc{\hn}{{\hat{n}}}

\def\diag{\mathop{\rm diag}}
\def\dim{\mathop{\rm Dim}}

\def\max{\mathop{\rm max}}

\def\dg{\dagger}

\def\ox{\otimes}

\newcommand{\bra}[1]{\langle#1|}
\newcommand{\ket}[1]{|#1\rangle}
\newcommand{\proj}[1]{| #1\rangle\!\langle #1 |}

\newcommand{\abs}[1]{|#1|}

\newcommand{\tbc}{\red{TO BE CONTINUED...}}

\newcommand{\opp}{\red{OPEN PROBLEMS}.~}

\newcommand{\red}{\textcolor{red}}

\def\Dbar{\leavevmode\lower.6ex\hbox to 0pt
	{\hskip-.23ex\accent"16\hss}D}

\begin{document}

\title{Mutually unbiased bases containing a complex Hadamard matrix of Schmidt rank three}
	
	\author{Mengyao Hu}\email[]{mengyaohu@buaa.edu.cn}
	\affiliation{School of Mathematical Sciences, Beihang University, Beijing 100191, China}

	\author{Yize Sun}\email[]{sunyize@buaa.edu.cn(corresponding author)}
	\affiliation{School of Mathematical Sciences, Beihang University, Beijing 100191, China}
	
\author{Lin Chen}\email[]{linchen@buaa.edu.cn(corresponding author)}
\affiliation{School of Mathematical Sciences, Beihang University, Beijing 100191, China}
\affiliation{International Research Institute for Multidisciplinary Science, Beihang University, Beijing 100191, China}

\date{\today}



\begin{abstract}
Constructing four six-dimensional mutually unbiased bases (MUBs) is an open problem in quantum physics and measurement. We investigate the existence of four MUBs including the identity, and
a complex Hadamard matrix (CHM) of Schmidt rank three. The CHM is equivalent to a controlled unitary operation on the qubit-qutrit system via local unitary transformation $I_2\otimes V$ and $I_2\otimes W$. We show that $V$ and $W$ have no zero entry, and apply it to exclude examples as members of MUBs. We further show that the maximum of entangling power of controlled unitary operation is $\log_2 3$ ebits. We derive the condition under which the maximum is achieved, and construct concrete examples. Our results describe the phenomenon that if a CHM of Schmidt rank three belongs to an MUB then its entangling power may not reach the maximum.
\end{abstract}
	
	\maketitle
	
	

\section{Introduction}
\label{sec:int}

The existence of four six-dimensional mutually unbiased bases (MUBs) is one of main open problems in quantum mechanics and information \cite{Schwinger1960UNITARY,Ivonovic1981Geometrical,WOOTTERS1989Optimal}. The problem is equivalent to showing the existence of identity matrix, and three $6\times6$ complex Hadamard matrices (CHMs) satisfying certain constraint. The CHM is a unitary matrix with entries of idental modulus. By regarding the column vectors of each CHM as an orthonormal basis in the six-dimensional Hilbert space $\bbC^6$, the constraint says that every two vectors from different bases has inner product of modulus $1/\sqrt6$. We shall denote an MUB trio as the set of three CHMs with above constraint, though it is widely believed that the set does not exist \cite{bw08,bw09,jmm09,bw10,deb10,wpz11,rle11,mw12jpa102001,mpw16,Goyeneche13}. MUBs play a key role in quantum tomography, key distribution, error correction, uncertainty relation, and more quantum correlations. The incompatibility of MUBs has been quantified using the noise robustness \cite{PhysRevLett.122.050402}. Recently, the MUB problem has been investigated using the extensively useful notion in quantum information, i.e., Schmidt rank \cite{Chen2017Product}, \footnote{Given a bipartite unitary operation $U$ on $\bbC^m\otimes\bbC^n$, its Schmidt rank is the integer $k$ such that $U=\sum^k_{j=1}A_j\otimes B_j$ with linearly independent $m\times m$ matrices $A_j$'s, and linearly independent $n\times n$ matrices $B_j$'s.}, see also Figure \ref{fig:sr}. It has been shown that the CHM with Schmidt rank one or two does not belong to any MUB trio. As far as we know, the approach of studying MUBs in terms of Schmidt rank is not much understood. Ref. \cite{Chen2017Product} displays a promising perspective on this long-standing problem. It is both physically meaningful and mathematically operational to investigate CHMs of larger Schmidt rank.

\begin{figure}
\centering
\includegraphics[width=0.50\textwidth]{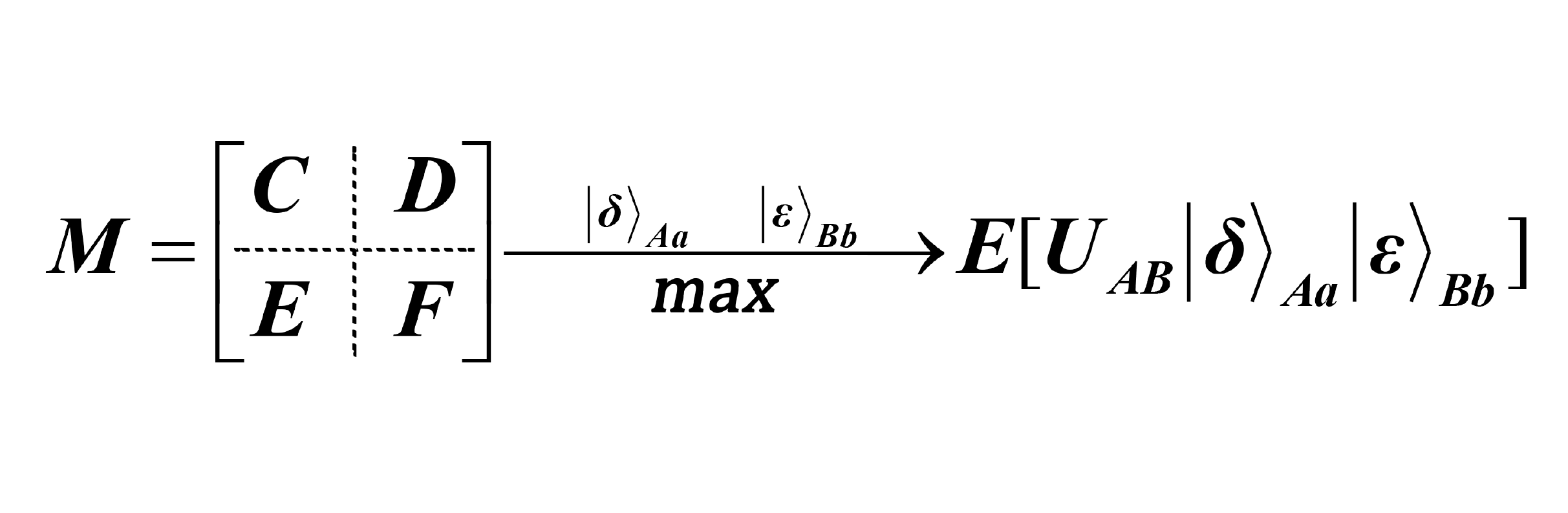}
\caption{The CHM $M$ consists of four blocks $C,D,E,F$. They are all $3\times3$ submatrices of entries of modulus $1/\sqrt6$. The Schmidt rank of $M$ is the number of linearly independent blocks in $C,D,E,F$. So the Schmidt rank is at most four. Recently it has been shown that if $M$ belongs to an MUB trio then $M$ has Schmidt rank three or four. We show that if $U_{AB}$ has Schmidt rank three then $M=(I_2\otimes V) U_{AB} (I_2 \otimes W)$ with $3\times3$ unitary matrices $V,W$ containing no zero entry, though numerical tests indicate that such $M$ may not exist. We further regard $M$ as a bipartite unitary operation, and investigate its entangling power. It is equal to the entangling power of $U_{AB}$, namely the maximum entanglement $E[U_{AB}\ket{\d}_{Aa}\ket{\e}_{Bb}]$ of the bipartite state $U_{AB}\ket{\d}_{Aa}\ket{\e}_{Bb}$ over all input states $\ket{\d}_{Aa}\ket{\e}_{Bb}$ with ancilla system $a,b$.}
\label{fig:sr}
\end{figure}

In this paper we shall investigate the CHM $M$ with Schmidt rank three. We review the preliminary results in Lemma \ref{le:sr2} and \ref{le:mubtrio}. We also construct examples of Schmidt-rank-three CHMs satisfying certain linear dependence in Lemma \ref{le:lineardependence=sr3}. Next, we characterize the expressions and properties of Schmidt-rank-three CHMs $M$ in Lemma \ref{le:sr3} and \ref{le:v:1zero}. It has been proven in \cite{cy14ap} that $M$ is a controlled unitary operation on $\bbC^2\otimes\bbC^3$, see Figure \ref{fig:chm=controlled}. So we obtain the decomposition $M=(I_2\otimes V) U_{AB} (I_2 \otimes W)$ with some $3\times3$ unitary matrices $V,W$, and $U_{AB}$ a controlled unitary operation controlled from the B side in the computational basis $\{\ket{j}\}$. Assisted by Lemma \ref{le:sr3} and \ref{le:v:1zero}, we show that if $V$ or $W$ has a zero entry then $M$ does not belong to any MUB trio in Theorem \ref{thm:sr3}. We show that some examples of $M$ do not belong to any MUB trio. It indicates that no Schmidt-rank-three CHM may belong to an MUB trio.
Since $M$ has Schmidt rank three, the maximum of entangling power of $U_{AB}$ is $\log_2 3$ ebits. In Eq. \eqref{eq:ujpo}, we analytically derive the condition under which the maximum is achieved. In Example \ref{ex:maxep=log23} we construct a concrete $U_{AB}$ by which the condition is satisfied. We also describe the lower bound of entangling power of general $U_{AB}$ in Figure \ref{fig:b3}, \ref{fig:b1} and \ref{fig:di}. In particular, the lower bound of a CHM as bipartite unitary operation may not reach the maximum entangling power, if the CHM belongs to an MUB trio. Our results show the potential connection between the open problem on the existence of four six-dimensional MUBs, and the entangling power of bipartite unitary operations in terms of Schmidt rank.

\begin{figure}
\center{\includegraphics{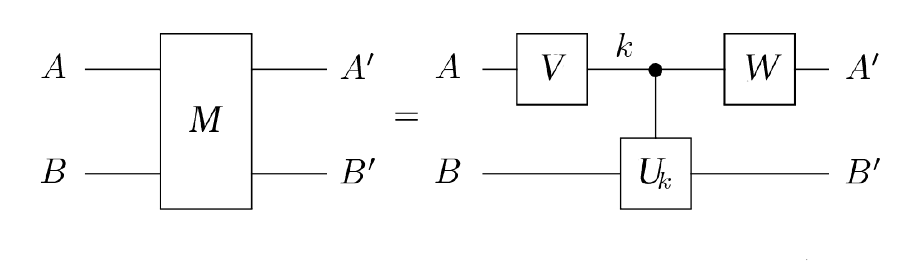}}
\caption{The CHM $M$ of Schmidt rank three is a bipartite controlled unitary operation on the space $\bbC^2\otimes\bbC^3$ controlled by system $B$. This is expressed as $M=(I\ox V)(\sum^3_{k=1}U_k \ox \proj{k})(I\ox W)$, where $V$, $U_k$, and $W$ are local unitary gates. The output systems $A'$ and $B'$ have the same size as that of $A$ and $B$, respectively. It implies that the CHM $M$ may be reliably implemented by experiments.}
\label{fig:chm=controlled}
\end{figure}

In quantum physics, the bipartite unitary operation is used for implementing quantum computing and cryptography. If the operation has Schmidt rank larger than one then it is nonlocal and can create entanglement. We ask for the maximum entanglement a nonlocal operation can create using a product state as an input state. The maximum is called the \textit{entangling power} of the nonlocal operation. The input state contains ancilla systems not affected by the operation. The study of entangling power has received extensive attentions in the past decades \cite{Nielsen03,lsw09,cy15,cy16,cy13,sm10}.
The entanglement power of a bipartite unitary operation is a lower bound of the entanglement required for realizing the operation under local operations and classical communications (LOCC). It is known that bipartite unitary operation with Schmidt rank at most three is a controlled unitary operation, and thus may be more easily implemented in experiments \cite{cy13,cy14ap}. Hence, studying the six-dimensional CHM in terms of Schmidt rank connects the MUB problem and bipartite unitary operations. Besides, the technique employed in our results relate the MUB problem to other fundamental notions like multiqubit entangled states, unextendible product basis and Birkhoff matrices \cite{Baerdemacker2017The,Chen2018The,kai=lma}. Very recently, bipartite operator of Schmidt rank three has been applied to the study of entanglement distillability of three bipartite reduced density operators from the same tripartite state \cite{1909.02748}.

The rest of this paper is structured as follows. In Sec. \ref{sec:pre} we introduce preliminary results from linear algebra and quantum information. In Sec. \ref{sec:result} we characterize the properties of Schmidt-rank-three CHMs. The proofs are given in Appendix \ref{app:sr3}, \ref{app:le:v:1zero} and \ref{app:thm 6 (iii)}. In Sec. \ref{sec:ep} we investigate the entangling power of Hadamard matrix, taken as a bipartite unitary operation on $\bbC^2\otimes\bbC^3$. We conclude in Sec. \ref{sec:con}.

\section{Preliminaries}
\label{sec:pre}

In this section we introduce the fundamental knowledge used throughout the paper. Let $\bbC^d$ be the $d$-dimensional Hilbert space, and $d=\prod^n_{j=1}d_j^{k_j}$ the prime factor decomposition such that $d_1^{k_1}<d_2^{k_2}<...<d_n^{k_n}$. It has been shown that there are at most $d+1$ and at least $d_1^{k_1}+1$ MUBs \cite{Ivonovic1981Geometrical,WOOTTERS1989Optimal}. In particular, the upper bound $d+1$ is achieved when $n=1$, namely $d$ is the prime power. On the other hand, if $n>1$ say $d=6$ then constructing MUBs becomes a hard problem.
The traditional way of studying the existence of MUBs employs Pauli groups, while we will do it using the Schmidt rank of CHMs. We say that two $mn\times mn$ matrices $A,B$ are equivalent if there exists a monomial unitary matrix $P\otimes Q$ and $R\otimes S$ with $P,R$ on $\bbC^m$ and $Q,S$ on $\bbC^n$, such that $(P\otimes Q)A(R\otimes S)=B$. If $A$ is a CHM then one can show that $B$ is also a CHM, and has the same Schmidt rank as that of $A$. The following result is from Lemma 13 of \cite{Chen2017Product}. It characterizes the expressions of order-six CHM of Schmidt rank up to three.
\bl
\label{le:sr2}
Any Schmidt-rank-three order-six CHM can be written as
\begin{widetext}
\bea
\label{eq:sr3}
&&
\bbH_3:=(I_2 \ox V)\cdot
\notag\\&&
\left(
                   \begin{array}{cccccc}
                     \cos\a_1 &  0 & 0 & e^{i\g_1}\sin\a_1 & 0 & 0 \\
                     0 &  \cos\a_2 &  0 & 0 & e^{i\g_2}\sin\a_2 & 0 \\
                     0 & 0 &  \cos\a_3 &  0 & 0 & e^{i\g_3}\sin\a_3 \\
                    e^{i\b_1}\sin\a_1 &  0 & 0 & -e^{i(\b_1+\g_1)}\cos\a_1 & 0 & 0 \\
                     0 &  e^{i\b_2}\sin\a_2 &  0 & 0 & -e^{i(\b_2+\g_2)}\cos\a_2 & 0 \\
                     0 & 0 &  e^{i\b_3}\sin\a_3 &  0 & 0 & -e^{i(\b_3+\g_3)}\cos\a_3 \\
                     \end{array}
                 \right)
\notag\\&&\cdot
(I_2 \ox W),
\eea
\end{widetext}
where $V$ and $W$ are
order-three unitary matrices, the first column vector of $W$ have all nonnegative and real elements, the matrix
\begin{eqnarray}
\label{eq:h4}
\bbH_4
:=
\bma
\cos\a_1 & e^{i\b_1}\sin\a_1 & e^{i\g_1}\sin\a_1
& -e^{i(\b_1+\g_1)}\cos\a_1\\
\cos\a_2 & e^{i\b_2}\sin\a_2 & e^{i\g_2}\sin\a_2
& -e^{i(\b_2+\g_2)}\cos\a_2\\
\cos\a_3 & e^{i\b_3}\sin\a_3 & e^{i\g_3}\sin\a_3
& -e^{i(\b_3+\g_3)}\cos\a_3\\
\ema	
\notag\\
\end{eqnarray}
has rank three by the parameters
$
\a_1,\a_2,\a_3\in[0,\p/2],~~~~\b_1,\b_2,\b_3,\g_1,\g_2,\g_3\in[0,2\p).
$ Hence (i.a) If $\bbH_3$ is a member of some MUB trio, then $\a_1,\a_2,\a_3\in(0,\p/2)$ and two of them are not equal.
\el

Recently, Ref. \cite{Chen2017Product} has shown that the CHM of Schmidt rank one or two does not belong to any MUB trio. So the next step is to treat CHMs of Schmidt rank three. Such CHMs exist. For example
\begin{eqnarray}
\label{eq:rank3chm}
{1\over\sqrt6}
\bma
y & -y & z & 1 & 1 & 1 \\
y\o^2 & -y\o^2 & z\o & 1 & 1 & \o \\
y\o & -y\o & z\o^2 & 1 & 1 & \o^2 \\
1 & 1 & 1 & x & -x & -z^* \\
1 & 1 & \o & x\o^2 & -x\o^2 & -z^*\o \\
1 & 1 & \o^2 & x\o & -x\o & -z^*\o^2 \\
\ema,
\end{eqnarray}
where $x,y,z$ are complex numbers of modulus one and ${z\over y}\ne {-z^*\over x}$. By applying Lemma \ref{le:sr2}, we construct more examples of Schmidt-rank-three CHMs satisfying certain linear dependence.
\begin{lemma}
\label{le:lineardependence=sr3}
Let $M$ be a Schmidt-rank-three order-six CHM whose four order-three submatrices are $A,B,C,D$. Then $M$ exists when one of the following two conditions is satisfied.

(i) Any three of $A,B,C,D$ are linearly independent. 	

(ii) $A,B,C$ are linearly dependent and any two of $A,B,C$ are linearly independent.	
\end{lemma}
\begin{proof}
(i) We choose $\a_i={\p\over4}$, $V={1\over\sqrt3}
\bma 1&1&1\\1&\o&\o^2\\1& \o^2 & \o \ema$, and $W=I_3$	in \eqref{eq:sr3}. Then $\bbH_3$ is an order-six CHM. The assertion is equivalent to finding $\b_i$ and $\g_i$ such that any order-three submatrix of $\bbH_4$ in \eqref{eq:h4} has rank three. An example is
$
\b_1=0,\b_2={\p\over6},\b_3={\p\over3},
\g_1=0,\g_2={\p\over3},\g_3={2\p\over3}.
$

(ii) We need find $\a_i,\b_i,\g_i$ such that the leftmost order-three submatrix $M$ of $\bbH_4$ has rank two, and any $3\times2$ submatrix of $M$ also has rank two. An example is $\a_1=\a_2=\a_3=1$, $\b_1=\b_2$, $\g_1=\g_2$, and $\b_3-\b_1\ne\g_3-\g_1$.
\end{proof}

Next we review the following observation from \cite[Lemma 11]{Chen2017Product}. It explains the necessary condition by which a $6\times6$ CHM is a member of some MUB trio. It will used frequently in the proofs for the claims in the next section. We shall refer to a \textit{subunitary matrix} as a matrix proportional to a unitary matrix.
\begin{lemma}
\label{le:mubtrio}
Any MUB trio contains no a real $3\times2$ or $2\times3$ real submatrix, or a $3\times3$ subunitary matrix.
\end{lemma}
For example, one can show that the CHM in \eqref{eq:rank3chm} does not belong to any MUB trio by Lemma \ref{le:mubtrio}.

\section{CHM of Schmidt rank three: MUB}
\label{sec:result}

In this section we characterize CHMs of Schmidt rank three. In Lemma \ref{le:sr3}, we investigate the cases when $V$ in \eqref{le:sr3} is a unitary matrix having exactly six and four zero entries, respectively. In Lemma \ref{le:v:1zero}, we investigate the cases when $V$ in \eqref{le:sr3} is a unitary matrix having exactly one zero entry. In Theorem \ref{thm:sr3} we present the main result of this section, namely the exclusion of CHM $\bbH_3$ with $V,W$ containing at least one zero entry.
We construct examples of $M$
and show that they do not belong to any MUB trio. It indicates that no Schmidt-rank-three CHM may belong to an MUB trio.

\begin{lemma}
\label{le:sr3}
	Let $\bbH_3$ be the Schmidt-rank-three order-six CHM in \eqref{eq:sr3}. We shall use the matrices $V,W$ and  parameters $\a_i,\b_i,\g_i$ in \eqref{eq:sr3}.
	
	(i) If $V$ is a monomial unitary matrix, then $
	\a_1=\a_2=\a_3=\p/4,~\b_1,\b_2,\b_3,\g_1,\g_2,\g_3\in[0,2\p)$, such that the matrix (\ref{eq:h4}) has rank three. Further
	$W=W_1D_1$ where
	$W_1=\frac{1}{\sqrt{3}}
	\begin{bmatrix}
	1 & 1 & 1\\
	1 & \omega & \omega^2\\
	1 & \omega^2 & \omega\\
	\end{bmatrix}$ or
	$\frac{1}{\sqrt{3}}
	\begin{bmatrix}
	1 & 1 & 1\\
	1 & \omega^2 & \omega\\
	1 & \omega & \omega^2\\
	\end{bmatrix}$, $D_1=\diag(1,e^{i\delta_2},e^{i\delta_3})$, and $\delta_2,\delta_3 \in [0,2\p)$.
	
	(ii) If $V$ is a unitary matrix with exactly four zero entries, then $V$ is equivalent to $1\oplus \frac{1}{\sqrt{2}}\bma
	1 & e^{i\t}\\
	1 & -e^{i\t}
	\ema$ where $\t=\pm\p/2$ when $\a_2\ne0,\p/2$. Next $
	\a_1=\p/4,~\a_2+\a_3=\p/2,$ $\a_2\in[0,\p/4)\cup(\p/4,\p/2]$, $~\b_1,\b_2,\b_3,\g_1,\g_2,\g_3\in[0,2\p)$, such that the matrix (\ref{eq:h4}) has rank three. Further
	$W=W_2D_2$, where $W_2=\frac{1}{\sqrt{3}}
	\begin{bmatrix}
	1 & 1 & 1\\
	1 & \omega & \omega^2\\
	1 & \omega^2 & \omega\\
	\end{bmatrix}$ or
	$\frac{1}{\sqrt{3}}
	\begin{bmatrix}
	1 & 1 & 1\\
	1 & \omega^2 & \omega\\
	1 & \omega & \omega^2\\
	\end{bmatrix}$, $D_2=\diag(1,e^{i\delta_2},e^{i\delta_3})$ and $\delta_2,\delta_3 \in [0,2\p)$.

\end{lemma}

The proof of this lemma is given in Appendix \ref{app:sr3}. Next we investigate the more complex case, namely $\bbH_3$ when $V$ in \eqref{eq:sr3} has exactly one zero entry.
	
\begin{lemma}
	\label{le:v:1zero}
	Let $\bbH_3$ be the Schmidt-rank-three order-six CHM in \eqref{eq:sr3}. If $V$ in \eqref{eq:sr3} is a unitary matrix with exactly one zero entry. Then $V$ is equivalent to $
	\bma
	v_{11} &v_{12} &0\\
	v_{21} &v_{22} &v_{23}\\
	v_{31} &v_{32} &v_{33}
	\ema$, where $v_{jk}\neq0$. Next $V$ is the product of two unitary matrices $p_1\oplus \frac{1}{\sqrt{2}}
	\bma
	1 &e^{i\theta}\\
	1 &-e^{i\theta}
	\ema$ and $\bma
	g_{21} &g_{23}\\
	g_{22} &g_{24}
	\ema\oplus p_2$, where $|g_{21}|^2=|g_{24}|^2=\frac{\cos2\a_2}{\cos2\a_2-\cos2\a_1}$. Third $\cos^2\a_1+\cos^2\a_2+\cos^2\a_3=\frac{3}{2},~\b_1,\b_2,\b_3,\g_1,\g_2,\g_3\in[0,2\p)$, such that the matrix (\ref{eq:h4}) has rank three. Further $W=W_3D_3$, where
	$W_3=
	\bma
	d_1 & e_1 & f_1\\
	d_2 & e_2 & f_2\\
	d_3 & e_3 & f_3
	\ema
	$, $d_1$, $d_2$ and $d_3$ are nonnegative and real numbers, and $D_3=\diag(1,e^{i\delta_2},e^{i\delta_3})$ is a diagonal unitary matrix, $\delta_2,\delta_3 \in [0,2\p)$.

\end{lemma}

The proof of this lemma is given in Appendix \ref{app:le:v:1zero}. Now we present the main result of this section.

\begin{theorem}
\label{thm:sr3}
Let $\bbH_3$ be the Schmidt-rank-three order-six CHM in \eqref{eq:sr3}.
	
(i) If $V$ is a monomial unitary matrix
then $\bbH_3$ is not a member of any MUB trio.

(ii) If $V$ is a unitary matrix with exactly four zero entries, then $\bbH_3$ is not a member of any MUB trio.
	
(iii) If $V$ is a unitary matrix with exactly one zero entry then $\bbH_3$ is not a member of any MUB trio. 	
\end{theorem}

\begin{proof}
(i) Since $V$ is a monomial unitary matrix, Lemma \ref{le:sr3} (i) shows that $\bbH_3$ has a $3\times3$ subunitary matrix. An example is the upper-left $3\times3$ submatrix of $\bbH_3$. This is the matrix $Y_3$ in Lemma \ref{le:mubtrio}. So assertion (i) holds.
	
(ii) Since $V=p\oplus G$ is a unitary matrix, Lemma \ref{le:sr3} (ii) shows that $\bbH_3$ has a $2\times3$ matrix of rank one. An example is the submatrix in the first and fourth rows, and the first three columns of $\bbH_3$.
This is the matrix $Y_2$ in Lemma \ref{le:mubtrio}. So assertion (ii) holds.	

(iii) The proof is based on Lemma \ref{le:v:1zero} and given in Appendix \ref{app:thm 6 (iii)}.
\end{proof}

By Theorem \ref{thm:sr3}, we conclude that the Schmidt-rank-three CHM $\bbH_3$ in \eqref{eq:sr3} in an MUB trio satisfies that the $3\times3$ unitary matrix $V$ in \eqref{eq:sr3} has no zero entry. Since $\bbH_3^\dg$ also belongs to an MUB trio, Theorem \ref{thm:sr3} shows that the $3\times3$ unitary matrix $W$ in \eqref{eq:sr3} has no zero entry too.
This fact shows that the CHM in the proof of Lemma \ref{le:lineardependence=sr3} (i) is excluded as a member of an MUB trio. Furthermore, the CHMs in \eqref{eq:rank3chm} and Lemma \ref{le:lineardependence=sr3} (ii) are excluded by Lemma \ref{le:mubtrio} and its full version in \cite{Chen2017Product}.
The above facts and Theorem \ref{thm:sr3} indicate that the MUB trio may not contain any CHM of Schmidt rank three.

On the other hand, the idea of constructing CHMs $M$ using the four blocks in Figure \ref{fig:sr} has been introduced by studying a four-parameter family of CHMs in \cite{Sz12}. It firstly determines the block $C$ of $M$, then finds out finitely many blocks $D$ and $E$, and finally determines whether the block $F$ exists. It is operational that we may choose suitable $D,E,F$ such that $M$ has Schmidt rank three. Furthermore, Ref. \cite{Sz12} constructs a four-parameter family of CHMs, and conjectures that it may be the full characterization of all CHMs. So the construction assisted by computer may bring about all CHMs of Schmidt rank three, especially $\bbH_3$ with $V,W$ having no zero entries.

\section{CHM of Schmidt rank three: entangling power}
\label{sec:ep}

In this section, we regard the six-dimensional CHM $\bbH_3$ in \eqref{eq:sr3} as a bipartite unitary operation on $\bbC^2\otimes\bbC^3$. It is known that $\bbH_3$ is a controlled unitary operation controlled from system $B$ \cite{cy14ap,cy14}, see Figure \ref{fig:chm=controlled}. We evaluate the entangling power of $\bbH_3$. Since the entangling power is invariant under local unitary transformation, we obtain that the entangling power of $\bbH_3$ is the same as that of
\begin{eqnarray}
\label{eq:uab}
U_{AB}=\sum^3_{k=1}U_k\otimes\proj{k},
\end{eqnarray}
where $U_k=\bma
\cos\a_k & e^{i\g_k}\sin\a_k \\
e^{i\b_k}\sin\a_k & -e^{i(\b_k+\g_k)}\cos\a_k \\
\ema $, and the real parameters $\a_j,\b_j,\g_j$ satisfy Eq. \eqref{eq:h4} and
$
\cos^2\a_1
+\cos^2\a_2
+\cos^2\a_3={3\over2}.
$
This equation is from the fact that $\bbH_3$ has entries of modulus $1/\sqrt6$.

Suppose $U_{AB}$ acts on the the input state, which is a bipartite product state $\ket{\d}_{Aa}\otimes\ket{\e}_{Bb}\in(\bbC^2\otimes\bbC^m)\otimes(\bbC^3\otimes\bbC^n)$, and $a,b$ are the ancilla systems. Up to local unitary transformation on system $a,b$ we may assume that
\begin{eqnarray}
\label{eq:alphaAa}
\ket{\d}_{Aa}=&&
(c_1\ket{1,1}+c_2\ket{1,2}+c_3\ket{2,2})_{Aa},
\\	
\label{eq:betaBa}
\ket{\e}_{Bb}=&&
(d_1\ket{1,x_1}+d_2\ket{2,x_2}+d_3\ket{3,x_3})_{Bb},
\end{eqnarray}
where
$
c_1,c_2,d_1,d_2,d_3\ge0,
c_1^2+c_2^2+\abs{c_3}^2=d_1^2+d_2^2+d_3^2=1.
$
The output state is the bipartite entangled state
\begin{eqnarray}
\label{eq:ps}	
\ket{\ps}_{Aa:Bb}=U_{AB}(\ket{\d}_{Aa}\otimes\ket{\e}_{Bb}).
\end{eqnarray}
By definition, the entangling power of $U_{AB}$ is the maximum entanglement of state $\ket{\ps}$ over all $\ket{\d}\otimes\ket{\e}$. It follows from \eqref{eq:uab} that $\ket{\ps}$ has Schmidt rank at most three. So it has the maximum entanglement $\log_2 3$ ebits. Using \eqref{eq:uab}-\eqref{eq:ps}, the maximum entanglement is achievable if and only if there exists $\ket{\d}_{Aa}$ such that the three states $(U_j)_A\ket{\d}_{Aa}$, $j=1,2,3$ are pairwise orthogonal. That is,
\begin{eqnarray}
&&
\bra{\d}_{Aa}(U_2)_A^\dg (U_1)_A	
\ket{\d}_{Aa}
\notag\\
=&&\bra{\d}_{Aa}(U_3)_A^\dg (U_2)_A	
\ket{\d}_{Aa}
\notag\\
= &&\bra{\d}_{Aa}(U_1)_A^\dg (U_3)_A	
\ket{\d}_{Aa}=0.
\label{eq:ujpo}
\end{eqnarray}
Using \eqref{eq:uab}-\eqref{eq:ujpo}, we can determine whether $U_{AB}$ has the maximum entanglement. A concrete example reaching the maximum is constructed as follows.
\begin{example}
\label{ex:maxep=log23}	
Let the input state be $\ket{\d}_{Aa}=\frac{\ket{11}+\ket{22}}{\sqrt{2}}$ by choosing $c_1=c_3=\frac{1}{\sqrt{2}},c_2=0$ in \eqref{eq:alphaAa}. Then \eqref{eq:ujpo} is equivalent to the statement that $\bbH_4 \bbH_4^\dg$ is diagonal, where $\bbH_4$ is from \eqref{eq:h4}. Let $U_{AB}$ in \eqref{eq:uab} satisfy $\b_1+\g_1=\b_2+\g_2=\b_3+\g_3+\p$, $\cos^2 \a_1+\cos^2\a_2=\frac{1}{2}$. By choosing $\a_1=\a_2=\p/3$ and $\cos(\b_1-\b_2)=-\frac{1}{3}$, we have
\begin{eqnarray}
&&
(\a_1,\a_2,\a_3,\b_1,\b_2,\b_3,\g_1,\g_2,\g_3)
\notag
\\=&&
		(\p/3,\p/3,0,\b_1,\b_1-\arccos(-\frac{1}{3}),\b_3,	
\notag\\&&
\g_1,\g_1+\arccos(-\frac{1}{3}),\b_1+\g_1-\b_3-\p).
\end{eqnarray}
By choosing good $\b_1$ and $\g_1$ we can obtain that $\bbH_4$ in \eqref{eq:h4} has rank three. For example, $\b_1=\g_1=0$ or $\b_1=\g_1=\p/2$. To conclude, we have shown that $U_{AB}$ corresponds to the six-dimensional CHM $\bbH_3$ in \eqref{eq:sr3}. Further, $U_{AB}$ is a Schmidt-rank-three bipartite unitary operation of maximum entangling power $\log_2 3$ ebits. Nevertheless, Lemma \ref{le:sr2} (i.a) and Theorem \ref{thm:sr3} show that such an $\bbH_3$ is not a member of any MUB trio.
\qed
\end{example}

In the remaining of this section, we investigate the entangling power of $\ket{\ps}$. Using \eqref{eq:ps}, one can obtain that the reduced density operator of system $Aa$ is the two-qubit state
\begin{eqnarray}
\label{eq:raa}
\r_{Aa}=&&
d_1^2 (U_1)_A \proj{\d}_{Aa} (U_1)_A^\dg	\notag\\
&&+
d_2^2 (U_2)_A \proj{\d}_{Aa} (U_2)_A^\dg	\notag\\
&&+
d_3^2 (U_3)_A \proj{\d}_{Aa} (U_3)_A^\dg	.
\end{eqnarray}
Evidently $\r_{Aa}$ has a zero eigenvalue. Suppose the three remaining eigenvalues are $\lambda_1,\lambda_2$ and $\lambda_3$. Then the entanglement of $\ket{\ps}$ is $S(\r_{Aa})=-\sum^3_{j=1}\lambda_j\log_2\lambda_j$, where $S(\r)$ is the von Neumann entropy of a quantum state $\r$. So the entangling power of $U_{AB}$ is
\begin{eqnarray}
\label{eq:max}
\max_{c_1,c_2,d_1,d_2,d_3\ge0,
\atop
c_1^2+c_2^2+\abs{c_3}^2=d_1^2+d_2^2+d_3^2=1}	
-\sum^3_{j=1}\lambda_j\log_2\lambda_j.
\label{eq:maxEP}
\end{eqnarray}
The analytical derivation of this maximum is out of reach yet, and we shall investigate its lower bound. We still use the parameters in Example \ref{ex:maxep=log23}, except that we replace $\a_1=\a_2=\p/3$ by $\p/3+x$ with $x \in [-\p/6,0]$. Correspondingly we replace $U_{AB}$ in Example \ref{ex:maxep=log23} by $U_{AB}(x)$. Note that the function $-\sum^3_{j=1}\lambda_j\log_2\lambda_j$ is a lower bound of the entangling power of $U_{AB}$ in \eqref{eq:max}.
 By using \eqref{eq:raa} we describe how the function changes with $\b_1,\b_3$ and $d_1,d_2,d_3$ in Figure \ref{fig:b3}-\ref{fig:di}. They imply that $-\sum^3_{j=1}\lambda_j\log_2\lambda_j$ has the maximum $\log_2 3$ only if $x=0$. Note that $U_{AB}(0)$ is exactly the bipartite unitary operation in Example \ref{ex:maxep=log23}. So $U_{AB}(0)$ reaches the maximum entangling power and does not belong to any MUB trio. The above pictures shows that the lower bound of entangling power of $U_{AB}(x)$ is smaller than that of $U_{AB}(0)$ as $x<0$. At the same time, we have not excluded the possibility that $U_{AB}(x)$ may belong to some MUB trio when $x<0$. It implies that the entangling power of a CHM may not reach the maximum, if it belongs to some MUB trio. It shows the potential connection between the existence of four MUBs and entangling power.
\begin{figure}[H]
	\centering
	\includegraphics[width=0.50\textwidth]{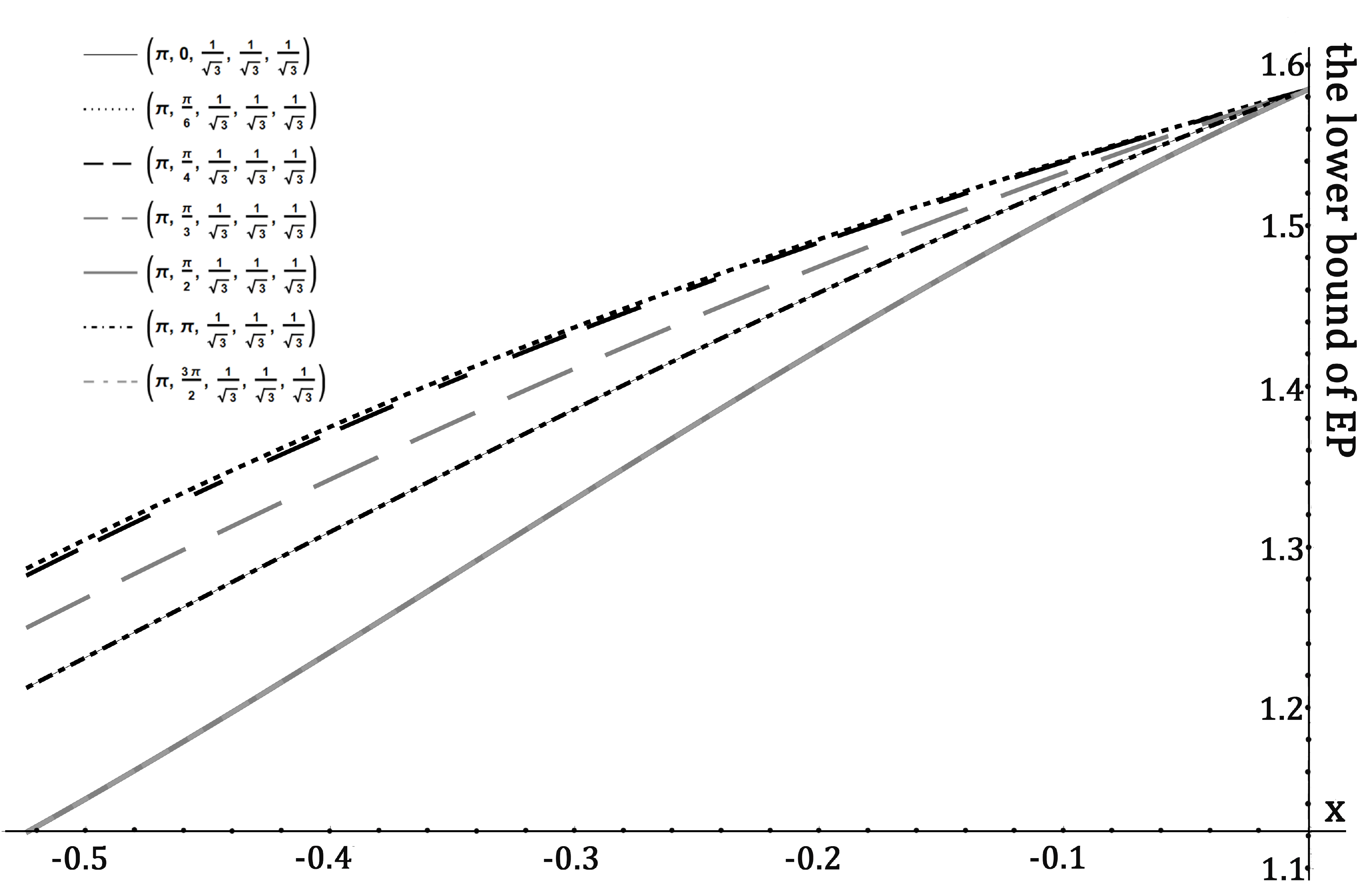}
	\caption{Let $(\b_1,\b_3,d_1,d_2,d_3)=(\p,\b_3,\frac{1}{\sqrt{3}},\frac{1}{\sqrt{3}},\frac{1}{\sqrt{3}})$, and $\b_3=0,\p/6,\p/4,\p/3,\p/2,\p,3\p/2$, respectively. The curves show that the function $-\sum^3_{j=1}\lambda_j\log_2\lambda_j$ in \eqref{eq:max} increases monotonically with $x\in[-\p/6,0]$. The curves $\b_3=0$ and $\b_3=\p$ coincide, and the curves of $\b_3=\p/2$ and $\b_3=3\p/2$ also coincide. The curve $\b_3$=$\p/6$ is above the curve $\b_3=\p/4$, and the curve $\b_3$=$\p/4$ is above the curve $\b_3=\p/3$.}
	\label{fig:b3}
\end{figure}

\begin{figure}[H]
	\centering
	\includegraphics[width=0.50\textwidth]{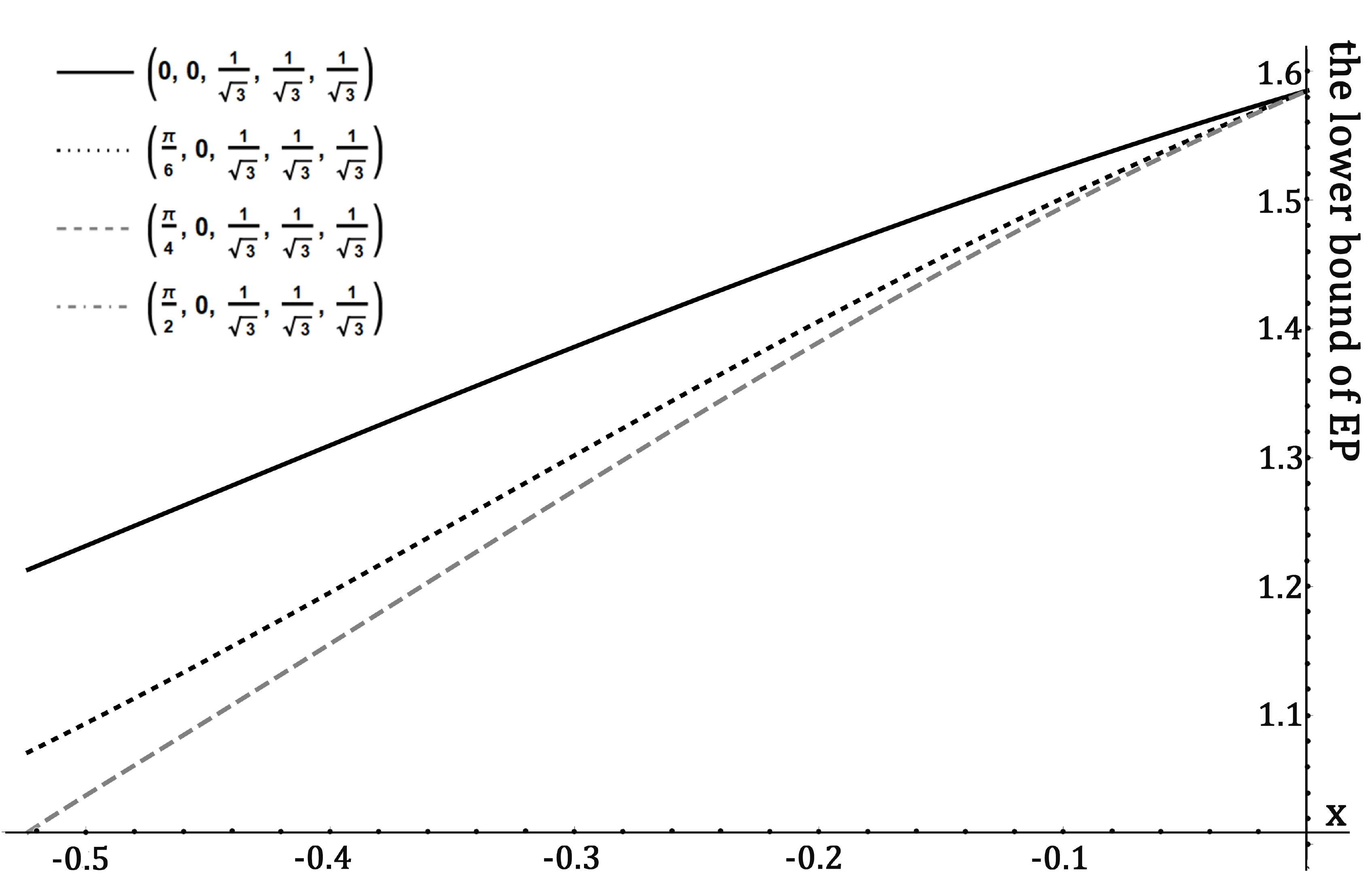}
	\caption{Let $(\b_1,\b_3,d_1,d_2,d_3)=(\b_1,0,\frac{1}{\sqrt{3}},\frac{1}{\sqrt{3}},\frac{1}{\sqrt{3}})$, and $\b_1=0,\p/6,\p/4,\p/2$, respectively.  The curves show that $-\sum^3_{j=1}\lambda_j\log_2\lambda_j$ increases monotonically with $x\in[-\p/6,0]$. The curves $\b_1=0$ and $\b_1=\p/2$ coincide. The curve  $\b_1$=$0$ is above the curve $\p/6$, and the curve $\b_1$=$0$ is above the curve $\b_1=\p/6$.  }
	\label{fig:b1}
\end{figure}

\begin{figure}[H]
	\centering
	\includegraphics[width=0.50\textwidth]{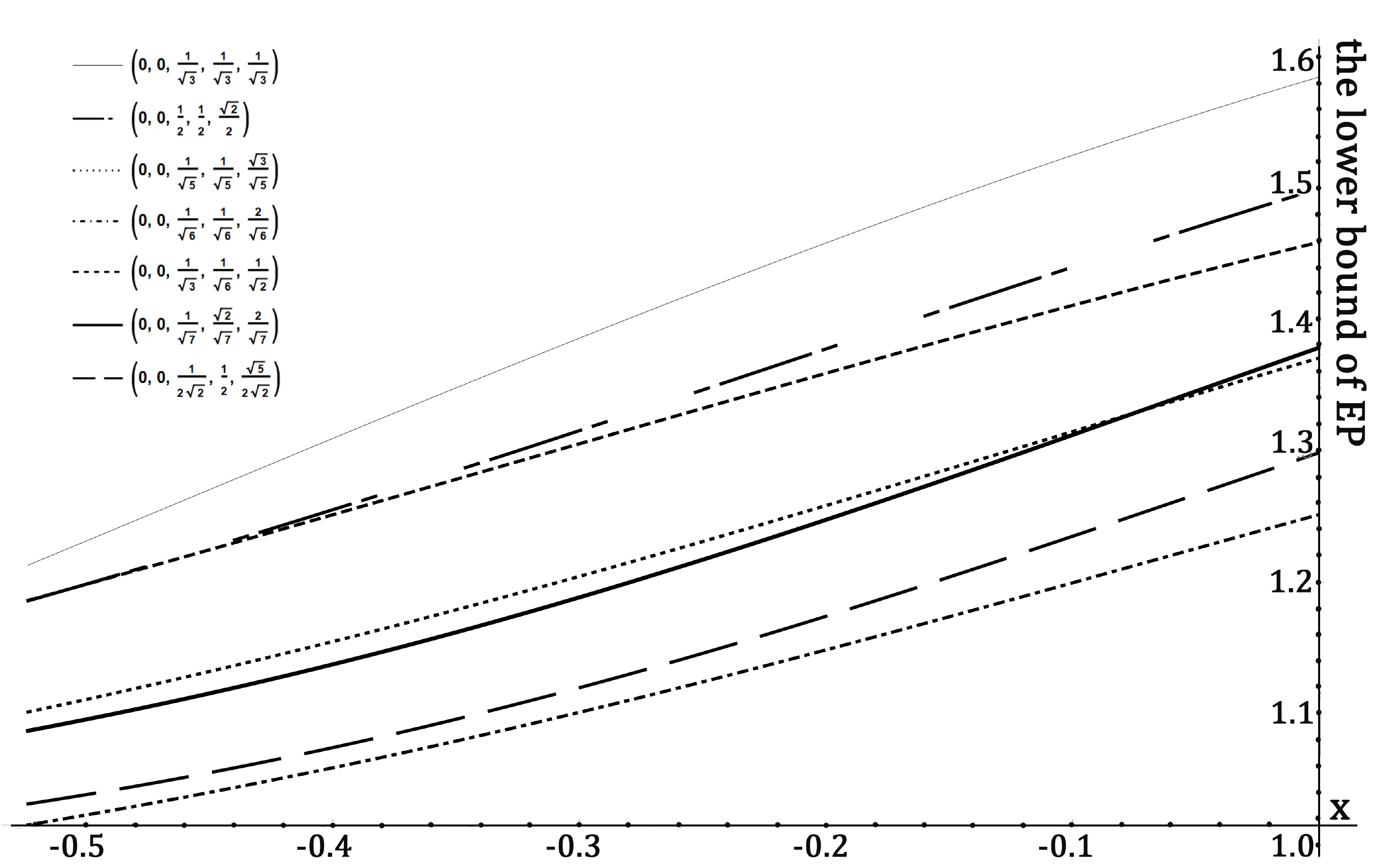}
	\caption{Let $(\b_1,\b_3,d_1,d_2,d_3)=(0,0,d_1,d_2,d_3)$, and $(d_1,d_2,d_3)=(\frac{1}{\sqrt{3}},\frac{1}{\sqrt{3}},\frac{1}{\sqrt{3}})$, $(\frac{1}{2},\frac{1}{2},\frac{\sqrt{2}}{2})$, $(\frac{1}{\sqrt{5}},\frac{1}{\sqrt{5}},\frac{\sqrt{3}}{\sqrt{5}})$, $(\frac{1}{\sqrt{6}},\frac{1}{\sqrt{6}},\frac{2}{\sqrt{6}})$, $(\frac{1}{\sqrt{3}},\frac{1}{\sqrt{6}},\frac{1}{\sqrt{2}})$, $(\frac{1}{\sqrt{7}},\frac{\sqrt{2}}{\sqrt{7}},\frac{2}{\sqrt{7}})$, $(\frac{1}{2\sqrt{2}},\frac{1}{2},\frac{\sqrt{5}}{2\sqrt{2}})$, respectively. The curves show that $-\sum^3_{j=1}\lambda_j\log_2\lambda_j$ increases monotonically with $x\in[-\p/6,0]$. And  $-\sum^3_{j=1}\lambda_j\log_2\lambda_j$ has the maximum $\log_23$ when $x=0$ and $(d_1,d_2,d_3)=(\frac{1}{\sqrt{3}},\frac{1}{\sqrt{3}},\frac{1}{\sqrt{3}})$.}
	\label{fig:di}
\end{figure}

\section{Conclusions}
	\label{sec:con}

The existence of four six dimensional MUBs consisting of an identity matrix and three CHMs has been a fundamental problem for decades. We have excluded a large subset of CHMs of Schmidt rank three from the four MUBs, and apply it to exclude examples constructed in this paper. It imposes a strict constraint on the existence of four six-dimensional MUBs, and this is supported by numerical tests.
We also have constructed the condition by which the entangling power of CHM as a bipartite controlled unitary operation is achieved. Our results indicate the conjecture that if a CHM of Schmidt rank three belongs to an MUB then its entangling power may not reach the maximum. The next target is to prove this conjecture, and analytically exclude any CHM of Schmidt rank three as a member of four MUBs including the identity matrix.

	\section*{Acknowledgments}
	\label{sec:ack}	
Authors were supported by the  NNSF of China (Grant No. 11871089), and the Fundamental Research Funds for the Central Universities (Grant Nos. KG12080401 and ZG216S1902).

\appendix

\section{Proof of Lemma \ref{le:sr3}}
\label{app:sr3}
	
\begin{proof}
	(i) Since $V$ is a monomial unitary matrix, there exists some unitary $D$ such that $V'=D V$ is a diagonal unitary matrix. So $V'$ is equivalent to $V$. Suppose
	$V'=
	\bma
	a & 0 & 0\\
	0 & b & 0\\
	0 & 0 & c
	\ema
	$,where $a$, $b$ and $c$ have modulus one. Suppose $W=W_1D_1$, where
	$W_1=
	\bma
	d_1 & e_1 & f_1\\
	d_2 & e_2 & f_2\\
	d_3 & e_3 & f_3
	\ema
	$, $d_1$, $d_2$ and $d_3$ are nonnegative and real numbers, and $D_1=\diag(1,e^{i\delta_2},e^{i\delta_3})$ is a diagonal unitary matrix, $\delta_2,\delta_3 \in [0,2\p)$. Then (\ref{eq:sr3}) is equivalent to
	\begin{widetext}
	\bea
	&&
	\bbH_3=(I_2 \ox V')(I_2 \ox P_1)\cdot
	\notag\\&&
	\left(
	\begin{array}{cccccc}
		\cos\a_1 &  0 & 0 & e^{i\g_1}\sin\a_1 & 0 & 0 \\
		0 &  \cos\a_2 &  0 & 0 & e^{i\g_2}\sin\a_2 & 0 \\
		0 & 0 &  \cos\a_3 &  0 & 0 & e^{i\g_3}\sin\a_3 \\
		e^{i\b_1}\sin\a_1 &  0 & 0 & -e^{i(\b_1+\g_1)}\cos\a_1 & 0 & 0 \\
		0 &  e^{i\b_2}\sin\a_2 &  0 & 0 & -e^{i(\b_2+\g_2)}\cos\a_2 & 0 \\
		0 & 0 &  e^{i\b_3}\sin\a_3 &  0 & 0 & -e^{i(\b_3+\g_3)}\cos\a_3 \\
	\end{array}
	\right)
	\notag\\&&
	\cdot
	(I_2 \ox W), \label{h3:1+1+1}
	\eea
	\end{widetext}
where $P_1$ is an arbitrary $3\times 3$ permutation matrix.

	In (\ref{h3:1+1+1}), let $\bbH_3=[h_{jk}]$, and $|h_{jk}|=\frac{1}{\sqrt{6}}$.
	Then we can obtain
	\begin{eqnarray}
		&&
		|ad_1\cos\a_1|=|ae_1\cos\a_1|=|af_1\cos\a_1|=\frac{1}{\sqrt{6}}, \notag\\
		\\&&
		|bd_2\cos\a_2|=|be_2\cos\a_2|=|bf_2\cos\a_2|=\frac{1}{\sqrt{6}},
		\\&&
		|cd_3\cos\a_3|=|ce_3\cos\a_3|=|cf_3\cos\a_3|=\frac{1}{\sqrt{6}}.	
	\end{eqnarray}
	So we have
	\begin{eqnarray}
		&&
		\cos\a_i=\frac{1}{\sqrt{6}d_i},
		\\&&
		|e_i|=|f_i|=d_i,\quad i=1,2,3.
	\end{eqnarray}
	By using $W_1W_1^{\dagger}=W_1^{\dagger}W_1=I$, we know
	\begin{eqnarray}
	   d_id_j+e_ie_j^{\ast}+f_if_j^{\ast}=\left\{
	   \begin{aligned}
	   &1 \quad \quad&i=j,\\
	   &0 \quad \quad&i\neq j.\\
	   \end{aligned}
	   \right.
	\end{eqnarray}
	They imply
	\begin{eqnarray}
     	&&
		|e_i|=|f_i|=d_i=\frac{1}{\sqrt{3}},
		\\&&
		a_i=\p/4,~\b_i,~\g_i\in[0,2\p),
		\\&&
		e_i,f_i\in{\dfrac{1}{\sqrt{3}}}\{1,\omega,\omega^2\},
		\\&&
		e_i\neq e_j, ~f_i\neq f_j.
	\end{eqnarray}		
	So we have 	$W_1=\frac{1}{\sqrt{3}}
	\begin{bmatrix}
	1 & 1 & 1\\
	1 & \omega & \omega^2\\
	1 & \omega^2 & \omega\\
	\end{bmatrix}$ or
	$\frac{1}{\sqrt{3}}
	\begin{bmatrix}
	1 & 1 & 1\\
	1 & \omega^2 & \omega\\
	1 & \omega & \omega^2\\
	\end{bmatrix}$.

	(ii) If $V$ is a unitary matrix which has four zero entries, then there exist permutation matrices $P_2, P_3$ such that $P_2VP_3=p\oplus G$, where  $G=
	\bma
	g_1 & g_3\\
	g_2 & g_4
	\ema$ is a non-monomial unitary matrix. Since $I_2 \ox P_3$ only changes its right matrix by row permutations, (\ref{eq:sr3}) is equivalent
	\begin{widetext}
	\bea
	&&
	\bbH_3=
	(I_2 \ox (p\oplus G))\cdot
	\notag\\&&
	\left(
	\begin{array}{cccccc}
		\cos\a_1 &  0 & 0 & e^{i\g_1}\sin\a_1 & 0 & 0 \\
		0 &  \cos\a_2 &  0 & 0 & e^{i\g_2}\sin\a_2 & 0 \\
		0 & 0 &  \cos\a_3 &  0 & 0 & e^{i\g_3}\sin\a_3 \\
		e^{i\b_1}\sin\a_1 &  0 & 0 & -e^{i(\b_1+\g_1)}\cos\a_1 & 0 & 0 \\
		0 &  e^{i\b_2}\sin\a_2 &  0 & 0 & -e^{i(\b_2+\g_2)}\cos\a_2 & 0 \\
		0 & 0 &  e^{i\b_3}\sin\a_3 &  0 & 0 & -e^{i(\b_3+\g_3)}\cos\a_3 \\
	\end{array}
	\right)
	\notag\\&&
	\cdot
	(I_2 \ox W),
	\label{h3:1+2}
	\eea
	\end{widetext}
up to the switching of $\a_j,\b_j$ and $\g_j$ and entries of $W$. Suppose $W=W_2D_2$, where
	$W_2=
	\bma
	d_1 & e_1 & f_1\\
	d_2 & e_2 & f_2\\
	d_3 & e_3 & f_3
	\ema
	$, $d_1$, $d_2$ and $d_3$ are nonnegative and real numbers, and $D_2=\diag(1,e^{i\delta_2},e^{i\delta_3})$ is a diagonal unitary matrix, $\delta_2,\delta_3 \in [0,2\p)$.
	
	Since $G$ is unitary we have $GG^\dagger=G^\dagger G=I$. Equivalently
	\begin{eqnarray}
		&&
		|g_2|^2=|g_3|^2=1-|g_1|^2=1-|g_4|^2, \label{eq:unit1}
		\\&&
		g_1g_2^\ast+g_3g_4^\ast=g_1^\ast g_3+g_2^\ast g_4=0. \label{eq:unit2}
	\end{eqnarray}
		
	By using $W_2W_2^{\dagger}=W_2^{\dagger}W_2=I$, we know
	\begin{eqnarray}
	d_id_j+e_ie_j^{\ast}+f_if_j^{\ast}=\left\{
	\begin{aligned}
	&1 \quad \quad&i=j,\\
	&0 \quad \quad&i\neq j.\\
	\end{aligned} \label{eq:unitW'}
	\right.
	\end{eqnarray}
	
	In (\ref{h3:1+2}), let $\bbH_3=[h_{jk}]$, and $|h_{jk}|=\frac{1}{\sqrt{6}}$.
	Then we can obtain
	\begin{eqnarray}
	    &&
		|pd_1\cos\a_1|=|pe_1\cos\a_1|=|pf_1\cos\a_1|=\dfrac{1}{\sqrt{6}}, \label{eq:mdl1}
		\\&&
		|pd_1e^{i\b_1}\sin\a_1|=|pe_1e^{i\b_1}\sin\a_1|=|pf_1e^{i\b_1}\sin\a_1|=\dfrac{1}{\sqrt{6}}, \label{eq:mdl2} \notag \\
		\\&&
		|pd_1e^{i\g_1}\sin\a_1|=|pe_1e^{i\g_1}\sin\a_1|=|pf_1e^{i\g_1}\sin\a_1|=\dfrac{1}{\sqrt{6}}, \label{eq:mdl-g} \notag \\
		\\&&
		|g_1d_2\cos\a_2+g_3d_3\cos\a_3|=|g_2d_2\cos\a_2+g_4d_3\cos\a_3|=\dfrac{1}{\sqrt{6}}, \label{eq:mdl-gd} \notag \\
		\\&&
		|g_1e_2\cos\a_2+g_3e_3\cos\a_3|=|g_2e_2\cos\a_2+g_4e_3\cos\a_3|=\dfrac{1}{\sqrt{6}}, \label{eq:mdl-ge} \notag \\
		\\&&
		|g_1f_2\cos\a_2+g_3f_3\cos\a_3|=|g_2f_2\cos\a_2+g_4f_3\cos\a_3|=\dfrac{1}{\sqrt{6}}. \label{eq:mdl-gf} \notag \\
	\end{eqnarray}	
	Simplifying (\ref{eq:mdl1}), (\ref{eq:mdl2}) and (\ref{eq:mdl-g}), we have
	\begin{eqnarray}
	    &&
		\a_1=\p/4, ~\b_i,~\g_i\in[0,2\p),
		\\&&
		d_1=\frac{1}{\sqrt{3}}, \label{eq:d1}
		\\&&
		|e_1|=|f_1|=\frac{1}{\sqrt{3}}.
	\end{eqnarray}	
	Using (\ref{eq:unit1}), (\ref{eq:unit2}) and (\ref{eq:mdl-gd}) we know
	\begin{eqnarray}
		&&
	    d_2^2\cos^2\a_2+d_3^2\cos^2\a_3=\dfrac{1}{3}, \label{eq:cosd}
	    \\&&
	    d_2^2\sin^2\a_2+d_3^2\sin^2\a_3=\dfrac{1}{3}.
	\end{eqnarray}
	They imply that
	\begin{eqnarray}
		&&
		d_2^2\cos2\a_2+d_3^2\cos2\a_3=0.	\label{eq:d2+d3}
	\end{eqnarray}
	Similarly, one can obtain
	\begin{eqnarray}
		&&
		|e_2|^2\cos^2\a_2+|e_3|^2\cos^2\a_3=\dfrac{1}{3}, \label{eq:cose}
		\\&&
		|e_2|^2\sin^2\a_2+|e_3|^2\sin^2\a_3=\dfrac{1}{3},
		\\&&
		|f_2|^2\cos^2\a_2+|f_3|^2\cos^2\a_3=\dfrac{1}{3}, \label{eq:cosf}
		\\&&
		|f_2|^2\sin^2\a_2+|f_3|^2\sin^2\a_3=\dfrac{1}{3}.
	\end{eqnarray}
	From (\ref{eq:cosd}), (\ref{eq:cose}) and (\ref{eq:cosf}), we have
	\begin{eqnarray}
		&&
		\bma
		d_2^2 &d_3^2 &-1/3\\
		|e_2|^2 &|e_3|^2 &-1/3\\
		|f_2|^2 &|f_3|^2 &-1/3
		\ema
		\bma
		\cos^2\a_2\\
		\cos^2\a_3\\
		1
		\ema=0. \label{eq:def}
	\end{eqnarray}
	Then (\ref{eq:unitW'}) and (\ref{eq:def}) imply that $\cos^2\a_2+\cos^2\a_3=1$. Note that $\a_2,\a_3\in[0,\p/2]$, so we have
	\begin{eqnarray}
		\a_2+\a_3=\p/2. \label{eq:a2+a3}
	\end{eqnarray}
	From (\ref{eq:d1}), (\ref{eq:d2+d3}) and (\ref{eq:a2+a3}), $d_2$ and $d_3$ are nonnegative and real numbers, we know $d_2=d_3=\dfrac{1}{\sqrt{3}}$. Similarly, we obtain
	\begin{eqnarray}
		&&
		|e_2|=|e_3|=|f_2|=|f_3|=\dfrac{1}{\sqrt{3}}.
	\end{eqnarray}
	Recall that the first and second row vectors of (\ref{h3:1+2}) are orthogonal, the first and third row vectors of (\ref{h3:1+2}) are orthogonal. We have
	\begin{eqnarray}
		&&
		e_1+e_2+e_3=0,
		\\&&
		f_1+f_2+f_3=0.
	\end{eqnarray}
	So we obtain
	\begin{eqnarray}
		&&
		(e_1,e_2,e_3)\propto (1,\o,\o^2) \text{\quad or \quad} (1,\o^2,\o),
		\\&&
		(f_1,f_2,f_3)\propto (1,\o,\o^2) \text{\quad or \quad} (1,\o^2,\o).
	\end{eqnarray}
	Hence $W_2=\frac{1}{\sqrt{3}}
	\begin{bmatrix}
	1 & 1 & 1\\
	1 & \omega & \omega^2\\
	1 & \omega^2 & \omega\\
	\end{bmatrix}$ or
	$\frac{1}{\sqrt{3}}
	\begin{bmatrix}
	1 & 1 & 1\\
	1 & \omega^2 & \omega\\
	1 & \omega & \omega^2\\
	\end{bmatrix}$.
	
	In the following we have two cases (ii.a) and (ii.b) in terms of $\a_2$.
	
	(ii.a) if $\a_2=\p/4$, then (\ref{eq:a2+a3}) implies that $\a_3=\p/4$. By using (\ref{eq:unit1}), (\ref{eq:unit2}) and (\ref{eq:mdl-gd}), we obtain
		$|g_1|^2+|g_3|^2=1,
		|g_1+g_3|=1.
	$
	They imply $g_1g_3=0$. It is a contradiction with the fact that $G$ is a non-monomial unitary matrix.
	
	(ii.b) if $\a_2\neq \p/4$, then from (\ref{eq:unitW'}) and (\ref{eq:mdl-gd})-(\ref{eq:mdl-gf}), we have
	\begin{eqnarray}
		&&
	    |g_1|^2\cos^2\a_2+|g_3|^2\cos^2\a_3=\frac{1}{2},
	    \\&&
	    |g_2|^2\cos^2\a_2+|g_4|^2\cos^2\a_3=\frac{1}{2}.
	\end{eqnarray}
	Using (\ref{eq:unit1}) and (\ref{eq:a2+a3}) we obtain
	\begin{eqnarray}
		&&
		|g_1|=|g_2|=|g_3|=|g_4|=\frac{1}{\sqrt{2}}. \label{mdl-gi}
	\end{eqnarray}
	Eqs. (\ref{eq:mdl-gd}) and (\ref{mdl-gi}) imply that
	\begin{eqnarray}
		&&
		\sin2\a_2(g_1g_3^\ast+g_1^\ast g_3)=0,
		\\&&
		\sin2\a_2(g_2g_4^\ast+g_2^\ast g_4)=0.
	\end{eqnarray}
	So we have two subcases (1) and (2).
	
	(1) If $\sin2\a_2=0$ then $\a_2=0$ or $\p/2$. Then we have $|g_1|=|g_2|=|g_3|=|g_4|=\frac{1}{\sqrt{2}}$. So $G$ is equivalent to $\frac{1}{\sqrt{2}}\bma
	1 & e^{i\t}\\
	1 & -e^{i\t}
	\ema$, where $\t \in [0,2 \p )$.
	
	(2) If $g_1g_3^\ast+g_1^\ast g_3=0$ and $g_2g_4^\ast+g_2^\ast g_4=0$, then we can obtain $g_1=g_2=\frac{1}{\sqrt{2}}$ and $g_3,g_4=\pm \frac{i}{\sqrt{2}}$. Namely $G=\frac{1}{\sqrt{2}}\bma
	1 & e^{i\t}\\
	1 & -e^{i\t}
	\ema$ where $\t=\pm\p/2$.
	
	We have proven assertion (ii).
\end{proof}

\section{Proof of Lemma \ref{le:v:1zero}}
\label{app:le:v:1zero}
	
\begin{proof}
From Lemma \ref{le:sr2} (i.a), we obtain that if $\bbH_3$ is a member of some MUB trio, then $\a_1,\a_2,\a_3\in(0,\p/2)$ and two of them are not equal.

	(i) If $V$ is a unitary matrix with exactly one zero entry, then there exist permutation matrices $P_4, P_5$ such that $V^{''}=P_4 V P_5=\bma
	v_{11} &v_{12} &0\\
	v_{21} &v_{22} &v_{23}\\
	v_{31} &v_{32} &v_{33}
	\ema$, where $v_{jk}\neq 0$. Since $I_2 \ox P_5$ only changes its right matrix by row permutations, (\ref{eq:sr3}) is equivalent
	\begin{widetext}
	\bea
	&&
	\bbH_3=
	(I_2 \ox V^{''})\cdot
	\notag\\&&
	\left(
	\begin{array}{cccccc}
		\cos\a_1 &  0 & 0 & e^{i\g_1}\sin\a_1 & 0 & 0 \\
		0 &  \cos\a_2 &  0 & 0 & e^{i\g_2}\sin\a_2 & 0 \\
		0 & 0 &  \cos\a_3 &  0 & 0 & e^{i\g_3}\sin\a_3 \\
		e^{i\b_1}\sin\a_1 &  0 & 0 & -e^{i(\b_1+\g_1)}\cos\a_1 & 0 & 0 \\
		0 &  e^{i\b_2}\sin\a_2 &  0 & 0 & -e^{i(\b_2+\g_2)}\cos\a_2 & 0 \\
		0 & 0 &  e^{i\b_3}\sin\a_3 &  0 & 0 & -e^{i(\b_3+\g_3)}\cos\a_3 \\
	\end{array}
	\right)
	\notag\\&&
	\cdot
	(I_2 \ox W),
	\label{h3:1+8}
	\eea
	\end{widetext}
	up to the switching of $\a_j,\b_j$ and $\g_j$ and entries of $W$. Suppose $W=W_3D_3$, where
	$W_3=
	\bma
	d_1 & e_1 & f_1\\
	d_2 & e_2 & f_2\\
	d_3 & e_3 & f_3
	\ema
	$, $d_1$, $d_2$ and $d_3$ are nonnegative and real numbers, and $D_3=\diag(1,e^{i\delta_2},e^{i\delta_3})$ is a diagonal unitary matrix, $\delta_2,\delta_3 \in [0,2\p)$.
	
	Since $V^{''}$ can be expressed as the product of two unitary matrices $p_1\oplus G_1$ and $G_2\oplus p_2$, where $G_1$ and $G_2$ are non-monomial unitary matrices. Namely
	\begin{eqnarray}
	V^{''}&&=(p_1\oplus G_1)\cdot (G_2 \oplus p_2) \notag
	\\&&
	=\bma
	p_1 &0 &0\\
	0 &g_{11} &g_{13}\\
	0 &g_{12} &g_{14}
	\ema \cdot
	\bma
	g_{21} &g_{23} &0\\
	g_{22} &g_{24} &0\\
	0 &0 &p_2
	\ema \notag \\
	&&=
	\bma
	p_1g_{21} &p_1g_{23} &0\\
	g_{11}g_{22} &g_{11}g_{24} &p_2g_{13}\\
	g_{12}g_{22} &g_{12}g_{24} &p_2g_{14}
	\ema \notag \\
	&&=
	\bma
	v_{11} &v_{12} &0\\
	v_{21} &v_{22} &v_{23}\\
	v_{31} &v_{32} &v_{33}
	\ema. \label{maV''}
	\end{eqnarray}
	By using $W_3W_3^{\dagger}=W_3^{\dagger}W_3=I$, we know
	\begin{eqnarray}
	d_id_j+e_ie_j^{\ast}+f_if_j^{\ast}=\left\{
	\begin{aligned}
	&1 \quad \quad&i=j,\\
	&0 \quad \quad&i\neq j.\\
	\end{aligned}
	\right. \label{unity:W3}
	\end{eqnarray}
	
	In (\ref{h3:1+8}), let $\bbH_3=[h_{jk}]$, and $|h_{jk}|=\frac{1}{\sqrt{6}}$.
	Then we can obtain
	\begin{eqnarray}
	&&
	|d_1v_{11}\cos\a_1+d_2v_{12}\cos\a_2| \notag \\
	&&=|e_1v_{11}\cos\a_1+e_2v_{12}\cos\a_2| \notag \\
	&&=|f_1v_{11}\cos\a_1+f_2v_{12}\cos\a_2|=\dfrac{1}{\sqrt{6}}, \label{cosa1+cosa2}
	\\&&
	|d_1v_{11}e^{i\b_1}\sin\a_1+d_2v_{12}e^{i\b_2}\sin\a_2| \notag \\
	&&=|e_1v_{11}e^{i\b_1}\sin\a_1+e_2v_{12}e^{i\b_2}\sin\a_2| \notag \\
	&&=|f_1v_{11}e^{i\b_1}\sin\a_1+f_2v_{12}e^{i\b_2}\sin\a_2|=\dfrac{1}{\sqrt{6}}, \label{eq:eib}
	\\&&
	|d_1v_{11}e^{i\g_1}\sin\a_1+d_2v_{12}e^{i\g_2}\sin\a_2| \notag \\
	&&=|e_1v_{11}e^{i\g_1}\sin\a_1+e_2v_{12}e^{i\g_2}\sin\a_2| \notag \\
	&&=|f_1v_{11}e^{i\g_1}\sin\a_1+f_2v_{12}e^{i\g_2}\sin\a_2|=\dfrac{1}{\sqrt{6}}, \label{eq:eig}
	\\&&
	|d_1v_{21}\cos\a_1+d_2v_{22}\cos\a_2+d_3v_{23}\cos\a_3| \notag \\
	&&=|d_1v_{31}\cos\a_1+d_2v_{32}\cos\a_2+d_3v_{33}\cos\a_3|=\dfrac{1}{\sqrt{6}}, \notag \\
	\label{unity:dv}
	\\&&
	|e_1v_{21}\cos\a_1+e_2v_{22}\cos\a_2+e_3v_{23}\cos\a_3| \notag \\
	&&=|e_1v_{31}\cos\a_1+e_2v_{32}\cos\a_2+e_3v_{33}\cos\a_3|=\dfrac{1}{\sqrt{6}}, \notag \\
	\label{unity:ev}
	\\&&
	|f_1v_{21}\cos\a_1+f_2v_{22}\cos\a_2+f_3v_{23}\cos\a_3| \notag \\
	&&=|f_1v_{31}\cos\a_1+f_2v_{32}\cos\a_2+f_3v_{33}\cos\a_3|=\dfrac{1}{\sqrt{6}}, \notag \\
	\label{unity:fv}
	\\&&
	|d_1v_{11}e^{i(\b_1+\g_1)}\cos\a_1+d_2v_{12}e^{i(\b_2+\g_2)}\cos\a_2| \notag \\
	&&=|e_1v_{11}e^{i(\b_1+\g_1)}\cos\a_1+e_2v_{12}e^{i(\b_2+\g_2)}\cos\a_2| \notag \\
	&&=|f_1v_{11}e^{i(\b_1+\g_1)}\cos\a_1+f_2v_{12}e^{i(\b_2+\g_2)}\cos\a_2|=\dfrac{1}{\sqrt{6}}. \notag \\
	\label{unity:vei(b+g)}
	\end{eqnarray}
	Using \eqref{maV''} and (\ref{cosa1+cosa2}) we obtain
	\begin{eqnarray}
	&&
	|v_{11}|^2\cos^2\a_1+|v_{12}|^2\cos^2\a_2=\dfrac{1}{2}. \label{eq:v11+v12}
	\end{eqnarray}
	From \eqref{maV''}, (\ref{unity:W3})-(\ref{cosa1+cosa2}) and (\ref{unity:dv})-(\ref{unity:fv}), we have
	\begin{eqnarray}
	&&
	|v_{21}|^2\cos^2\a_1+|v_{22}|^2\cos^2\a_2+|v_{23}|^2\cos^2\a_3=\dfrac{1}{2}, \notag\\ \label{eq:v21+v22+v23}
	\\&&
	|v_{31}|^2\cos^2\a_1+|v_{32}|^2\cos^2\a_2+|v_{33}|^2\cos^2\a_3=\dfrac{1}{2},\notag\\ \label{eq:v31+v32+v33}
	\\&&
	d_1^2\cos^2\a_1+d_2^2\cos^2\a_2+d_3^2\cos^2\a_3=\dfrac{1}{2}, \label{dcos}
	\\&&
	|e_1|^2\cos^2\a_1+|e_2|^2\cos^2\a_2+|e_3|^2\cos^2\a_3=\dfrac{1}{2},\notag\\ \label{ecos}
	\\&&
	|f_1|^2\cos^2\a_1+|f_2|^2\cos^2\a_2+|f_3|^2\cos^2\a_3=\dfrac{1}{2}. \notag\\ \label{fcos}
	\end{eqnarray}
	
	Eqs. (\ref{unity:W3}) and (\ref{dcos})-(\ref{fcos}) imply that
	\begin{eqnarray}
	&&
	\cos^2\a_1+\cos^2\a_2+\cos^2\a_3=\frac{3}{2}. \label{cosa1+cosa2+cosa3}
	\end{eqnarray}
	
	Recall that the first and second column vectors of (\ref{h3:1+8}) are orthogonal, the first and third column vectors of (\ref{h3:1+8}) are orthogonal, the second and third column vectors of (\ref{h3:1+8}) are orthogonal. We have
	\begin{eqnarray}
	&&
	d_1e_1+d_2e_2+d_3e_3=0, \label{iii:di*ei}
	\\&&
	d_1f_1+d_2f_2+d_3f_3=0,
	\\&&
	e_1^\ast f_1+e_2^\ast f_2+e_3^\ast f_3=0. \label{iii:ei*fi}
	\end{eqnarray}
	
	By simplifying (\ref{maV''}) and (\ref{cosa1+cosa2})-(\ref{eq:eig}), we can obtain
	\begin{eqnarray}
	&&
	\sin\a_1\sin\a_2(g_{21}^\ast g_{23}(e^{i(\b_2-\b_1)}-e^{i(\g_2-\g_1)}) \notag
	\\&&
	+g_{21}g_{23}^\ast(e^{i(\b_2-\b_1)}-e^{i(\g_2-\g_1)})^\ast)=0, \label{iii:b2-b1,g2-g1}
	\\&&
	\sin\a_1\sin\a_2(e_1^\ast e_2g_{21}^\ast g_{23}(e^{i(\b_2-\b_1)}-e^{i(\g_2-\g_1)}) \notag
	\\&&
	+e_1e_2^\ast g_{21}g_{23}^\ast(e^{i(\b_2-\b_1)}-e^{i(\g_2-\g_1)})^\ast)=0, \label{iii:e1e2_b2-b1,g2-g1}
	\\&&
	\sin\a_1\sin\a_2(f_1^\ast f_2g_{21}^\ast g_{23}(e^{i(\b_2-\b_1)}-e^{i(\g_2-\g_1)}) \notag
	\\&&
	+f_1f_2^\ast g_{21}g_{23}^\ast(e^{i(\b_2-\b_1)}-e^{i(\g_2-\g_1)})^\ast)=0. \label{iii:f1f2_b2-b1,g2-g1}
	\end{eqnarray}
	
	Using (\ref{cosa1+cosa2}) and (\ref{unity:vei(b+g)}) we have
	\begin{eqnarray}
	&&
	\cos\a_1\cos\a_2(g_{21}^\ast g_{23}(e^{i(\b_2+\g_2-\b_1-\g_1)}-1) \notag
	\\&&
	+g_{21}g_{23}^\ast(e^{i(\b_2+\g_2-\b_1-\g_1)}-1)^\ast)=0,  \label{iii:d1d2_b2+g2-b1-g1}
	\\&&
	\cos\a_1\cos\a_2(e_1^\ast e_2g_{21}^\ast g_{23}(e^{i(\b_2+\g_2-\b_1-\g_1)}-1) \notag
	\\&&
	+e_1e_2^\ast g_{21}g_{23}^\ast(e^{i(\b_2+\g_2-\b_1-\g_1)}-1)^\ast)=0, \label{iii:e1e2_b2+g2-b1-g1}
	\\&&
	\cos\a_1\cos\a_2(f_1^\ast f_2g_{21}^\ast g_{23}(e^{i(\b_2+\g_2-\b_1-\g_1)}-1) \notag
	\\&&
	+f_1f_2^\ast g_{21}g_{23}^\ast(e^{i(\b_2+\g_2-\b_1-\g_1)}-1)^\ast)=0. \label{iii:f1f2_b2+g2-b1-g1}
	\end{eqnarray}
	
	From (\ref{maV''}) and (\ref{eq:v11+v12}) we have
	\begin{eqnarray}
	&&
	|g_{21}|^2=|g_{24}|^2=\frac{\cos2\a_2}{\cos2\a_2-\cos2\a_1}, \label{mdl:g21}
	\\&&
	|g_{22}|^2=|g_{23}|^2=\frac{-\cos2\a_1}{\cos2\a_2-\cos2\a_1}, \label{mdl:g22}
	\end{eqnarray}
	where $\a_1\in (0,\p/4), \a_2\in(\p/4,\p/2)$ or $\a_1\in(\p/4,\p/2), \a_2\in (0,\p/4)$.
	So $G_2=\bma
	g_{21} &g_{23}\\
	g_{22} &g_{24}
	\ema$ where $|g_{21}|^2=|g_{24}|^2=\frac{\cos2\a_2}{\cos2\a_2-\cos2\a_1}$ and $	|g_{22}|^2=|g_{23}|^2=\frac{-\cos2\a_1}{\cos2\a_2-\cos2\a_1}$. And there exist a diagonal unitary matrix $\diag (e^{i\vartheta_1},e^{i\vartheta_2},1)$ such that $\diag (e^{i\vartheta_1},e^{i\vartheta_2},1)(G_2\oplus 1)=\bma
	e^{i\vartheta_1}g_{21} & e^{i\vartheta_1}g_{23} &0\\
	e^{i\vartheta_2}g_{22} & e^{i\vartheta_2}g_{24} &0\\
	0 &0 &1
	\ema$, where $e^{i\vartheta_1}g_{21}$ and $e^{i\vartheta_2}g_{22}$ are real numbers, $\vartheta_1,\vartheta_2\in [0,2\p)$. Hence $G_2=\bma
	g_{21} &g_{23}\\
	g_{22} &g_{24}
	\ema$ is equivalent to $\bma
	\sqrt{\frac{\cos2\a_2}{\cos2\a_2-\cos2\a_1}} &e^{i\vartheta_1}g_{23}\\
	\sqrt{\frac{-\cos2\a_1}{\cos2\a_2-\cos2\a_1}} &e^{i\vartheta_2}g_{24}
	\ema$, where $|e^{i\vartheta_1}g_{23}|=\sqrt{\frac{-\cos2\a_1}{\cos2\a_2-\cos2\a_1}}$, $|e^{i\vartheta_2}g_{24}|=\sqrt{\frac{\cos2\a_2}{\cos2\a_2-\cos2\a_1}}$.
	
	Simplifying (\ref{eq:v21+v22+v23})-(\ref{eq:v31+v32+v33}) and (\ref{mdl:g21})-(\ref{mdl:g22}) we can obtain $|g_{11}|=|g_{12}|=|g_{13}|=|g_{14}|=\frac{1}{\sqrt{2}}$. So $G_1$ is equivalent to $\frac{1}{\sqrt{2}}
	\bma
	1 &e^{i\theta}\\
	1 &-e^{i\theta}
	\ema$ where $\theta\in[0,2\p]$.
\end{proof}	

\section{Proof of Theorem \ref{thm:sr3} (iii)}
\label{app:thm 6 (iii)}
\begin{proof}
	In Lemma \ref{le:v:1zero}, we characterize $\bbH_3$. Here we prove that
	if $V$ is a unitary matrix with exactly one zero entry then $\bbH_3$ is not a member of any MUB trio.
	
	From (\ref{iii:b2-b1,g2-g1})-(\ref{iii:f1f2_b2-b1,g2-g1}) and (\ref{iii:d1d2_b2+g2-b1-g1})-(\ref{iii:f1f2_b2+g2-b1-g1}), we know $\sin\a_1\sin\a_2=0$ or $\cos\a_1\cos\a_2=0$ is a contradiction with the condition $\a_1,\a_2,\a_3\in(0,\p/2)$ and two of them are not equal. So we have
	\begin{widetext}
	\begin{eqnarray}
	&&
	g_{21}^\ast g_{23}(e^{i(\b_2-\b_1)}-e^{i(\g_2-\g_1)})+g_{21}g_{23}^\ast(e^{i(\b_2-\b_1)}-e^{i(\g_2-\g_1)})^\ast=0,
	\\&&
	e_1^\ast e_2g_{21}^\ast g_{23}(e^{i(\b_2-\b_1)}-e^{i(\g_2-\g_1)})
	+e_1e_2^\ast g_{21}g_{23}^\ast(e^{i(\b_2-\b_1)}-e^{i(\g_2-\g_1)})^\ast=0,
	\\&&
	f_1^\ast f_2g_{21}^\ast g_{23}(e^{i(\b_2+\g_2-\b_1-\g_1)}-1)
	+f_1f_2^\ast g_{21}g_{23}^\ast(e^{i(\b_2+\g_2-\b_1-\g_1)}-1)^\ast=0,
	\\&&
	g_{21}^\ast g_{23}(e^{i(\b_2+\g_2-\b_1-\g_1)}-1)
	+g_{21}g_{23}^\ast(e^{i(\b_2+\g_2-\b_1-\g_1)}-1)^\ast=0,
	\\&&
	e_1^\ast e_2g_{21}^\ast g_{23}(e^{i(\b_2+\g_2-\b_1-\g_1)}-1)
	+e_1e_2^\ast g_{21}g_{23}^\ast(e^{i(\b_2+\g_2-\b_1-\g_1)}-1)^\ast=0,
	\\&&
	f_1^\ast f_2g_{21}^\ast g_{23}(e^{i(\b_2+\g_2-\b_1-\g_1)}-1)
	+f_1f_2^\ast g_{21}g_{23}^\ast(e^{i(\b_2+\g_2-\b_1-\g_1)}-1)^\ast=0.
	\end{eqnarray}
	\end{widetext}
		
	They imply the following two cases (i) and (ii).
	
	(i) If $g_{21}^\ast g_{23}(e^{i(\b_2-\b_1)}-e^{i(\g_2-\g_1)})$, $e_1^\ast e_2g_{21}^\ast g_{23}(e^{i(\b_2-\b_1)}-e^{i(\g_2-\g_1)})$, $f_1^\ast f_2g_{21}^\ast g_{23}(e^{i(\b_2-\b_1)}-e^{i(\g_2-\g_1)})$, $g_{21}^\ast g_{23}(e^{i(\b_2+\g_2-\b_1-\g_1)}-1)$, $	e_1^\ast e_2g_{21}^\ast g_{23}(e^{i(\b_2+\g_2-\b_1-\g_1)}-1)$ and $f_1^\ast f_2g_{21}^\ast g_{23}(e^{i(\b_2+\g_2-\b_1-\g_1)}-1)$ are all pure imaginaries. Then $e_1^\ast e_2$ and $f_1^\ast f_2$ are real numbers. Evidently there exists a diagonal unitary $D_3'=\diag(1,e^{i\zeta_1},e^{i\zeta_2})$ such that $W_3D_3'=
	\bma
	d_1 &e^{i\zeta_1}e_1 &e^{i\zeta_2}f_1\\
	d_2 &e^{i\zeta_1}e_2 &e^{i\zeta_2}f_2\\
	d_3 &e^{i\zeta_1}e_3 &e^{i\zeta_2}f_3
	\ema=
	\bma
	d_1 &e'_1 &f'_1\\
	d_2 &e'_2 &f'_2\\
	0 &e'_3 &f'_3
	\ema$ is a unitary where $e'_1,e'_2,f'_1,f'_2$ are real, $e'_3,f'_3$ are complex, $\zeta_1,\zeta_2 \in [0,2\p]$. Then from (\ref{iii:di*ei})-(\ref{iii:ei*fi}), we can obtain $e_3=f_3=0$ or $d_3=0$.
	
	If $e_3=f_3=0$ then $W$ in (\ref{h3:1+8}) is equivalent to the case $V=1\oplus G$ in Lemma \ref{le:sr3}. Since if $I_6, B_1, B_2~and~B_3$ are four MUBs in six-dimensional system then there exists $B_1^\dagger$ such that $I_6, B_1, B_2~and~B_3$ become $B_1^\dagger, I_6, B_1^\dagger B_2~and~B_1^\dagger B_3$, respectively. So $B_1^\dagger, I_6, B_1^\dagger B_2~and~B_1^\dagger B_3$ are still four MUBs. Hence if $e_3=f_3=0$ then $\bbH_3^\dagger$ is equivalent to the case $V=1\oplus G$ in Lemma \ref{le:sr3}.
	
	If $d_3=0$, then from (\ref{dcos}), $W_3W_3^{\dagger}=W_3^{\dagger}W_3=I$, $d_1$ and $d_2$ are nonnegative and real numbers, one can obtain
	\begin{eqnarray}
	&&
	d_1=\sqrt{\dfrac{\cos2\a_2}{\cos2\a_2-\cos2\a_1}}, \label{iii:d1}
	\\&&
	d_2=\sqrt{\dfrac{-\cos2\a_1}{\cos2\a_2-\cos2\a_1}}. \label{iii:d2}
	\end{eqnarray}
	Then from (\ref{iii:di*ei})-(\ref{iii:ei*fi}), we know
	\begin{eqnarray}
	&&
	\frac{e_1}{e_2}=-\sqrt{\frac{-\cos2\a_1}{\cos2\a_2}}, \label{iii:e1/e2}
	\\&&
	\frac{f_1}{f_2}=-\sqrt{\frac{-\cos2\a_1}{\cos2\a_2}}.
	\end{eqnarray}
	And \eqref{cosa1+cosa2} becomes
	\begin{eqnarray}
	&&
	|d_1g_{21}\cos\a_1+d_2g_{23}\cos\a_2|=|e'_1g_{21}\cos\a_1+e'_2g_{23}\cos\a_2| \notag \\
	&&=|f'_1g_{21}\cos\a_1+f'_2g_{23}\cos\a_2|=\dfrac{1}{\sqrt{6}}.    \label{iii:e1'e2'}
	\end{eqnarray}
	Eqs. (\ref{iii:d1})-(\ref{iii:e1'e2'}) imply that $(e'_1,e'_2)=\pm(f'_1,f'_2)$. Let $(e'_1,e'_2)=(f'_1,f'_2)$, so $W_3D_3'$ becomes $\bma
	d_1 &e'_1 &e'_1\\
	d_2 &e'_2 &e'_2\\
	0 &e'_3 &f'_3
	\ema$. Since the second column and the third column vectors of $W_3D_3'$ are orthorgonal, we know $e'_3=-f'_3$, $|e'_3|=|e_3|=\frac{1}{\sqrt{2}}$, $e'^2_1+e'^2_2=\frac{1}{2}$.
	Then from (\ref{cosa1+cosa2+cosa3}) and (\ref{iii:e1/e2}), we have
	\begin{eqnarray}
	&&
	|e'_1|=\frac{1}{\sqrt{2}}\sqrt{\frac{-\cos2\a_1}{\cos2\a_2-\cos2\a_1}},
	\\&&
	|e'_2|=\frac{1}{\sqrt{2}}\sqrt{\frac{\cos2\a_2}{\cos2\a_2-\cos2\a_1}}.
	\end{eqnarray}
	So $W_3D'_3=\bma
	d_1 & \frac{1}{\sqrt{2}}d_2 &\frac{1}{\sqrt{2}}d_2\\
	d_2 &-\frac{1}{\sqrt{2}}d_1 &-\frac{1}{\sqrt{2}}d_1\\
	0 &\frac{1}{\sqrt{2}}e^{i\zeta_1} &-\frac{1}{\sqrt{2}}e^{i\zeta_1}
	\ema$, where $d_1=\sqrt{\dfrac{\cos2\a_2}{\cos2\a_2-\cos2\a_1}},	d_2=\sqrt{\dfrac{-\cos2\a_1}{\cos2\a_2-\cos2\a_1}}$, $\zeta_1\in [0,2\p]$.
	
	From (\ref{mdl:g21})-(\ref{mdl:g22}), we can obtain that $G_2$ is equivalent to $
	\bma
	d_1 &d_2e^{i\vartheta_1}\\
	d_2 &-d_1e^{i\vartheta_1}
	\ema$, where $d_1=\sqrt{\dfrac{\cos2\a_2}{\cos2\a_2-\cos2\a_1}},	d_2=\sqrt{\dfrac{-\cos2\a_1}{\cos2\a_2-\cos2\a_1}}$, $\vartheta_1\in[0,2\p)$.
	Hence
	\begin{eqnarray}
	&&
	V''=(p_1\oplus G_1)\diag(e^{-i\vartheta_1},e^{-i\vartheta_1},1)
	\\&&\diag(e^{i\vartheta_1},e^{i\vartheta_1},1)(G_2\oplus p_2) \notag
	\\&&
	=\bma
	d_1 &d_2 &0\\
	\frac{1}{\sqrt{2}}d_2 &-\frac{1}{\sqrt{2}}d_1 &\frac{1}{\sqrt{2}}\\
	\frac{1}{\sqrt{2}}d_2 &-\frac{1}{\sqrt{2}}d_1 &-\frac{1}{\sqrt{2}}
	\ema
	\cdot
	\bma
	e^{-i\vartheta_1} &0 &0\\
	0& 1 &0\\
	0& 0&e^{i\theta}
	\ema,
	\\&&
	W=(W_3D'_3)D'^{-1}_3D_3 \notag
	\\&&
	=
	\bma
	1 &0 &0\\
	0 &1 &0\\
	0 &0 &e^{i\zeta_1}
	\ema
	\cdot
	\bma
	d_1 & \frac{1}{\sqrt{2}}d_2 &\frac{1}{\sqrt{2}}d_2\\
	d_2 &-\frac{1}{\sqrt{2}}d_1 &-\frac{1}{\sqrt{2}}d_1\\
	0 &\frac{1}{\sqrt{2}} &-\frac{1}{\sqrt{2}}
	\ema	\cdot
	\\&&
	\bma
	1 &0 &0\\
	0 &e^{-i\zeta_1} &0\\
	0 &0& e^{-i\zeta_1}
	\ema
	\cdot
	\bma
	1 &0 &0\\
	0 &e^{i\delta_1} &0\\
	0 &0& e^{i\delta_2}
	\ema,
	\end{eqnarray}
	where $d_1=\sqrt{\dfrac{\cos2\a_2}{\cos2\a_2-\cos2\a_1}},	d_2=\sqrt{\dfrac{-\cos2\a_1}{\cos2\a_2-\cos2\a_1}}$, $\vartheta_1, \theta, \zeta_1, \delta_1,\delta_2 \in[0,2\p)$.
	Then (\ref{h3:1+8}) is equivalent to
	\begin{widetext}
	\begin{eqnarray}
	&&
	\bbH_3
	=(I_2 \ox
	(\bma
	d_1 &d_2 &0\\
	\frac{1}{\sqrt{2}}d_2 &-\frac{1}{\sqrt{2}}d_1 &\frac{1}{\sqrt{2}}\\
	\frac{1}{\sqrt{2}}d_2 &-\frac{1}{\sqrt{2}}d_1 &-\frac{1}{\sqrt{2}}
	\ema
	\bma
	e^{-i\vartheta_1} &0 &0\\
	0& 1 &0\\
	0& 0&e^{i\vartheta_2}
	\ema)) \cdot
	\notag\\&&
	\left(
	\begin{array}{cccccc}
	\cos\a_1 &  0 & 0 & e^{i\g_1}\sin\a_1 & 0 & 0 \\
	0 &  \cos\a_2 &  0 & 0 & e^{i\g_2}\sin\a_2 & 0 \\
	0 & 0 &  \cos\a_3 &  0 & 0 & e^{i\g_3}\sin\a_3 \\
	e^{i\b_1}\sin\a_1 &  0 & 0 & -e^{i(\b_1+\g_1)}\cos\a_1 & 0 & 0 \\
	0 &  e^{i\b_2}\sin\a_2 &  0 & 0 & -e^{i(\b_2+\g_2)}\cos\a_2 & 0 \\
	0 & 0 &  e^{i\b_3}\sin\a_3 &  0 & 0 & -e^{i(\b_3+\g_3)}\cos\a_3 \\
	\end{array}
	\right)
	\notag\\&&
	\cdot
	(I_2 \ox
	\bma
	d_1 & \frac{1}{\sqrt{2}}d_2 &\frac{1}{\sqrt{2}}d_2\\
	d_2 &-\frac{1}{\sqrt{2}}d_1 &-\frac{1}{\sqrt{2}}d_1\\
	0 &\frac{1}{\sqrt{2}} &-\frac{1}{\sqrt{2}}
	\ema),	\label{h3:1+8equivalent}
	\end{eqnarray}
	where  $\vartheta_1,\vartheta_2 \in[0,2\p]$ are the functions of $\a_1, \a_2$.
	Next we compute $\vartheta_1,\vartheta_2$ such that (\ref{h3:1+8equivalent}) is a Schmidt-rank-three order-six CHM.
	In (\ref{h3:1+8equivalent}), $|h_{jk}|=\frac{1}{\sqrt{6}}$, so we have	
	\begin{eqnarray}
	&&
	|d_1^2e^{-i\vartheta_1}\cos \a_1+d_2^2\cos\a_2|=\frac{1}{\sqrt{2}}|d_1d_2e^{-i\vartheta_1}\cos\a_1-d_1d_2\cos\a_2|=\frac{1}{\sqrt{6}}, \label{iii:VAW1}
	\\&&
	\frac{1}{2}|d_2^2e^{-i\vartheta_1}\cos\a_1+d_1^2\cos\a_2+e^{i\vartheta_2}\cos\a_3|
	=\frac{1}{2}|d_2^2e^{-i\vartheta_1}\cos\a_1+d_1^2\cos\a_2-e^{i\vartheta_2}\cos\a_3|=\frac{1}{\sqrt{6}}, \label{iii:VAW2}
	\\&&
	|d_1^2e^{i(\g_1-\vartheta_1)}+d_2^2e^{i\g_2}\sin\a_2|
	=\frac{1}{\sqrt{2}}|d_1d_2e^{i(\g_1-\vartheta_1)}\sin\a_1-d_1d_2e^{i\g_2}\sin\a_2|=\frac{1}{\sqrt{6}}, \label{iii:VBW1}
	\\&&
	\frac{1}{2}|d_2^2e^{i(\g_1-\vartheta_1)}+d_1^2e^{i\g_2}\sin\a_2+e^{i(\g_3+\vartheta_2)}\sin\a_3|
	=\frac{1}{2}|d_2^2e^{i(\g_1-\vartheta_1)}+d_1^2e^{i\g_2}\sin\a_2-e^{i(\g_3+\vartheta_2)}\sin\a_3|=\frac{1}{\sqrt{6}},  \label{iii:VBW2}
	\\&&
	|d_1^2e^{i(\b_1-\vartheta_1)}\sin\a_1+d_2^2e^{i\b_2}\sin\a_2|
	\frac{1}{\sqrt{2}}|d_1d_2e^{i(\b_1-\vartheta_1)}\sin\a_1-d_1d_2e^{i\b_2}\sin\a_2|=\frac{1}{\sqrt{6}},    \label{iii:VCW1}
	\\&&
	\frac{1}{2}|d_2^2e^{i(\b_1-\vartheta_1)}+d_1^2e^{i\b_2}\sin\a_2+e^{i(\b_3+\vartheta_2)}\sin\a_3|
	=\frac{1}{2}|d_2^2e^{i(\b_1-\vartheta_1)}+d_1^2e^{i\b_2}\sin\a_2-e^{i(\b_3+\vartheta_2)}\sin\a_3|=\frac{1}{\sqrt{6}},   \label{iii:VCW2}
	\\&&
	|d_1^2e^{i(\b_1+\g_1-\vartheta_1)}+d_2^2e^{i(\b_2+\g_2)}\cos\a_2|
	=\frac{1}{\sqrt{2}}|d_1d_2e^{i(\b_1+\g_1-\vartheta_1)}\cos\a_1-d_1d_2e^{i(\b_2+\g_2)}\cos\a_2|=\frac{1}{\sqrt{6}},    \label{iii:VDW1}
	\\&&
	\frac{1}{2}|d_2^2e^{i(\b_1+\g_1-\vartheta_1)}\cos\a_1+d_1^2e^{i(\b_2+\g_2)}\cos\a_2+e^{i(\b_3+\g_3+\vartheta_2)}\cos\a_3| \notag \\
	&&
	=\frac{1}{2}|d_2^2e^{i(\b_1+\g_1-\vartheta_1)}\cos\a_1+d_1^2e^{i(\b_2+\g_2)}\cos\a_2-e^{i(\b_3+\g_3+\vartheta_2)}\cos\a_3|=\frac{1}{\sqrt{6}}. \label{iii:VDW2}
	\end{eqnarray}
	
	Simplifying (\ref{iii:VAW2}) and (\ref{iii:VDW2}) one can obtain
	\begin{eqnarray}
	&&
	\cos(\vartheta_1+\vartheta_2)\cos\a_1\cos2\a_1=\cos\vartheta_2\cos\a_2\cos2\a_2, \notag \label{iii:coscos}
	\\&&
	\cos(\vartheta_1+\vartheta_2+\b_3+\g_3-\b_1-\g_1)\cos\a_1\cos2\a_1  \notag
	\\&&
	=\cos(\vartheta_2+\b_3+\g_3-\b_2-\g_2)\cos\a_2\cos2\a_2.
	\end{eqnarray}
	Similarly, from (\ref{iii:VBW2}) and (\ref{iii:VCW2}), we have
	\begin{eqnarray}
	&&
	\cos(\vartheta_1+\vartheta_2-\g_3-\g_1)\sin\a_1\cos2\a_1=\cos(\vartheta_2+\g_3-\g_2)\sin\a_2\cos2\a_2,
	\\&&
	\cos(\vartheta_1+\vartheta_2+\b_3-\b_1)\sin\a_1\cos2\a_1=\cos(\vartheta_2+\b_3-\b_2)\sin\a_2\cos2\a_2. \label{iii:sincos}
	\end{eqnarray}
By Comparing (\ref{iii:VAW1}) and (\ref{iii:VDW1}), (\ref{iii:VBW1}) and (\ref{iii:VCW1}), we obtain
	\begin{eqnarray}
	&&
	\cos(\vartheta_1+\b_2+\g_2-\b_1-\g_1)=\cos\vartheta_1,
	\\&&
	\cos(\vartheta_1+\g_2-\g_1)=\cos(\vartheta_1+\b_2-\b_1).
	\end{eqnarray}
	They imply
	\begin{eqnarray}
	&&
	\vartheta_1+\frac{\g_2+\b_2-\g_1-\b_1}{2}=k_1\p ~~or~~\frac{\g_2+\b_2-\g_1-\b_1}{2}=k_2\p,   \label{iii:s1}
	\\&&
	\vartheta_1+\frac{\g_2+\b_2-\g_1-\b_1}{2}=k_1\p  ~~or~~\frac{\g_2-\g_1-(\b_2-\b_1)}{2}=k_3\p, \label{iii:s2}
	\end{eqnarray}
	where $k_1,k_2,k_3 \in \bbZ$. Thus in the following we have two subcases (i.a) and (i.b) in terms of (\ref{iii:s1}) and (\ref{iii:s2}).
	
	(i.a). If $\frac{\g_2+\b_2-\g_1-\b_1}{2}=k_2\p$ and $\frac{\g_2-\g_1-(\b_2-\b_1)}{2}=k_3\p$, then $\g_2-\g_1=(k_2+k_3)\p$, $\b_2-\b_1=(k_2-k_3)\p$, $s=(k_1-k_3)\p$. Then (\ref{iii:coscos})-(\ref{iii:sincos}) can be simplified as
	\begin{eqnarray}
	&&
	\frac{\cos\a_2\cos 2\a_2}{\cos\a_1\cos2\a_2}=\pm 1,
	\\&&
	\frac{\sin\a_2\cos 2\a_2}{\sin\a_1\cos2\a_2}=\pm 1.
	\end{eqnarray}
	So we have
	\begin{eqnarray}
	&&
	\frac{\cos\a_2\sin\a_1}{\cos\a_1\sin\a_2}=\pm 1.
	\end{eqnarray}
	It implies that $\a_1=\a_2$ or $\a_1+\a_2=\p$. It is a contradiction with the condition $\a_1\in (0,\p/4), \a_2\in(\p/4,\p/2)$ or $\a_1\in(\p/4,\p/2), \a_2\in (0,\p/4)$.
	
	(i.b) If $\vartheta_1+\frac{\g_2+\b_2-\g_1-\b_1}{2}=k_1\p$. From (\ref{iii:VAW1}) and (\ref{iii:coscos}), we have
	\begin{eqnarray}
	&&
	\cos\vartheta_1=\frac{-(\cos 2\a_1-\cos 2\a_2)^2+6\cos^2 2\a_1 \cos^2\a_2+6\cos^2\a_1 \cos^2 2\a_2}{12\cos\a_1 \cos 2\a_1 \cos \a_2 \cos 2\a_2}, \label{expr:cos s}
	\\&&
	\cos^2 \vartheta_2=\frac{(-2+\cos 2\a_1)^2\sec ^2\a_2+8\cos 2\a_1\sec 2\a_2(4+\cos 2\a_1(2+\sec 2\a_2))}{24(1+3\cos 2\a_1+3\cos 2\a_2)}. \label{expr:cos t2} \notag
	\\
	\end{eqnarray}
	Note that in (\ref{h3:1+8equivalent}),
	\begin{eqnarray}
	&&
	(I_2 \ox
	\bma
	e^{-i\vartheta_1} &0 &0\\
	0& 1 &0\\
	0& 0&e^{i\vartheta_2}
	\ema) \cdot
	\notag\\&&
	\left(
	\begin{array}{cccccc}
	\cos\a_1 &  0 & 0 & e^{i\g_1}\sin\a_1 & 0 & 0 \\
	0 &  \cos\a_2 &  0 & 0 & e^{i\g_2}\sin\a_2 & 0 \\
	0 & 0 &  \cos\a_3 &  0 & 0 & e^{i\g_3}\sin\a_3 \\
	e^{i\b_1}\sin\a_1 &  0 & 0 & -e^{i(\b_1+\g_1)}\cos\a_1 & 0 & 0 \\
	0 &  e^{i\b_2}\sin\a_2 &  0 & 0 & -e^{i(\b_2+\g_2)}\cos\a_2 & 0 \\
	0 & 0 &  e^{i\b_3}\sin\a_3 &  0 & 0 & -e^{i(\b_3+\g_3)}\cos\a_3 \\
	\end{array}
	\right) \notag
	\end{eqnarray}
	is unitary.
	We have
	\begin{eqnarray}
	&&
	1+e^{2i(\b_1+\g_1)}
	=1+e^{-2i(\b_2+\g_2)}
	=1+e^{2i(\b_3+\g_3)}.
	\end{eqnarray}
	It implies that $\b_1+\g_1=(m_1+\frac{1}{2})\p,\b_2+\g_2=(m_2+\frac{1}{2})\p,\b_3+\g_3=(m_3+\frac{1}{2})\p$, where $m_1,m_2,m_3\in \bbZ$. Together with \eqref{iii:coscos}-\eqref{iii:sincos}, one can obtain
	\begin{eqnarray}
	\vartheta_1&&=k_1\p+\frac{\g_1+\b_1-\g_2-\b_2}{2}=\frac{1}{2}(2 k_1+m_1-m_2)\p,
	\\
	\vartheta_2&&=\frac{1}{2}(\b_2+\g_2-\b_3-\g_3-2k_1 \p +k_2 \p)=\frac{1}{2}(-2 k_1+k_2+m_2-m_3)\p.
	\end{eqnarray}
	Hence $\vartheta_1+\vartheta_2=\frac{1}{2}(k_2+m_1-m_3)\p$. It means that $\frac{\cos(\vartheta_1+\vartheta_2)}{\cos\vartheta_2}=\frac{\cos\a_2\cos2\a_2}{\cos\a_1\cos2\a_1}$ is equal to -1,0 or 1. So $(\cos (\vartheta_1+\vartheta_2),\cos \vartheta_2)=(1,-1), (-1,1)$ or $(0,0)$. In the first or second case we have $\cos \vartheta_1=-1$. From the expressions of $\cos\vartheta_1$ and $\cos^2 \vartheta_2$ in (\ref{expr:cos s})-(\ref{expr:cos t2}), one can obtain that $(\cos\vartheta_1,\cos^2 \vartheta_2)=(-1,1)$ has no solutions with $\a_1\in (0,\p/4), \a_2\in(\p/4,\p/2)$ or $\a_1\in(\p/4,\p/2), \a_2\in (0,\p/4)$. Similarly, the third case is excluded by $(\cos\vartheta_1,\cos^2 \vartheta_2)=(0,0)$ from the expressions in (\ref{expr:cos s})-(\ref{expr:cos t2}).

	(ii) If $\b_2-\b_1-(\g_2-\g_1)=0$ and $\b_2+\g_2-\b_1-\g_1=0$ then $\b_2=\b_1$, $\g_2=\g_1$. Then \eqref{eq:sr3} is equivalent to
	\begin{eqnarray}
	&&
	(I_2 \otimes
	\bma
	e^{i \vartheta_2}\sqrt{\frac{\cos2\a_2}{\cos2\a_2-\cos2\a_1}} &-e^{i \vartheta_2}\sqrt{\frac{-\cos2\a_1}{\cos2\a_2-\cos2\a_1}} &0\\
	\frac{1}{\sqrt{2}} \sqrt{\frac{-\cos2\a_1}{\cos2\a_2-\cos2\a_1}} &  \sqrt{\frac{\cos2\a_2}{\cos2\a_2-\cos2\a_1}} & \frac{1}{\sqrt{2}} \\
	\frac{1}{\sqrt{2}} \sqrt{\frac{-\cos2\a_1}{\cos2\a_2-\cos2\a_1}} &  \sqrt{\frac{\cos2\a_2}{\cos2\a_2-\cos2\a_1}} & -\frac{1}{\sqrt{2}}
	\ema) \cdot \notag \\
	&&
	\left(
	\begin{array}{cccccc}
	e^{-i \vartheta_2}\cos\a_1 &  0 & 0 & e^{i(\g_1-\vartheta_2)}\sin\a_1 & 0 & 0 \\
	0 &  \cos\a_2 &  0 & 0 & e^{i\g_1}\sin\a_2 & 0 \\
	0 & 0 &  e^{i \theta}\cos\a_3 &  0 & 0 & e^{i(\g_3+\theta)}\sin\a_3 \\
	e^{i(\b_1-\vartheta_2)}\sin\a_1 &  0 & 0 & -e^{i(\b_1+\g_1-\vartheta_2)}\cos\a_1 & 0 & 0 \\
	0 &  e^{i\b_1}\sin\a_2 &  0 & 0 & -e^{i(\b_1+\g_1)}\cos\a_2 & 0 \\
	0 & 0 &  e^{i(\b_3+\theta)}\sin\a_3 &  0 & 0 & -e^{i(\b_3+\g_3+\theta)}\cos\a_3 \\
	\end{array}
	\right) \cdot
	\notag \\
	&&
	(I_2 \otimes
	\bma
	d_1 & e_1 & f_1 \\
	d_2 & e_2 & f_2\\
	d_3 & e_3 & f_3
	\ema). \label{iii:ma:b1=b2,g1=g2}
	\end{eqnarray}
In \eqref{iii:ma:b1=b2,g1=g2}, let $|h_{jk}|=\frac{1}{\sqrt{6}}$. One can obtain
	\begin{eqnarray}
	&&
	|-d_2 e^{i \vartheta_2}\cos\a_2 \sqrt{\frac{-\cos2\a_1}{\cos2\a_2-\cos2\a_1}}+d_1\cos \a_1\sqrt{\frac{\cos2\a_2}{\cos2\a_2-\cos2\a_1}}|=\frac{1}{\sqrt{6}}, \label{iii:mdl:h_11}
	\\&&
	|e^{i \b_1}(d_1\sqrt{\frac{\cos2\a_2}{\cos2\a_2-\cos2\a_1}}\sin\a_1 -d_2 e^{i \vartheta_2}\sqrt{\frac{-\cos2\a_1}{\cos2\a_2-\cos2\a_1}} \sin\a_2 )|=\frac{1}{\sqrt{6}}, \label{iii:mdl:h_41}
	\\&&
	|-e^{i\vartheta_2}e_2 \cos\a_2\sqrt{\frac{-\cos2\a_1}{\cos2\a_2-\cos2\a_1}}+e_1\cos\a_1 \sqrt{\frac{\cos2\a_2}{\cos2\a_2-\cos2\a_1}} |=\frac{1}{\sqrt{6}}, \label{iii:mdl:h_12}
	\\&&
	|e^{i \b_1}(e_1\sqrt{\frac{\cos2\a_2}{\cos2\a_2-\cos2\a_1}} \sin\a_1-e^{i\vartheta_2}e_2\sqrt{\frac{-\cos2\a_1}{\cos2\a_2-\cos2\a_1}}\sin\a_2)|=\frac{1}{\sqrt{6}}. \label{iii:mdl:h_42}
	\end{eqnarray}
	Simplifying 	\eqref{iii:mdl:h_11} and \eqref{iii:mdl:h_41} we have
	\begin{eqnarray}
	&&
	d_1^2\frac{\cos 2\a_2}{\cos 2\a_2-\cos 2\a_1}+d_2^2\frac{-\cos 2\a_1}{\cos 2\a_2-\cos 2\a_1}-2d_1 d_2\sqrt{\frac{-\cos2\a_1}{\cos2\a_2-\cos2\a_1}}\sqrt{\frac{\cos2\a_2}{\cos2\a_2-\cos2\a_1}}\cos(\a_1-\a_2)\cos \vartheta_2=\frac{1}{3}, \notag
	\\
	\\&&
	d_1^2 \frac{\cos 2\a_1 \cos 2\a_2}{\cos 2\a_2-\cos 2\a_1}+d_2^2\frac{-\cos 2\a_1 \cos 2\a_2}{\cos 2\a_2-\cos 2\a_1}-2d_1 d_2\sqrt{\frac{-\cos2\a_1}{\cos2\a_2-\cos2\a_1}}\sqrt{\frac{\cos2\a_2}{\cos2\a_2-\cos2\a_1}}\cos(\a_1+\a_2)\cos \vartheta_2=0. \notag
	\\
	\end{eqnarray}	
	By solving the above equations, we obtain
	\begin{eqnarray}
	&&
	\cos \vartheta_2=\frac{\cos 2\a_2-\cos 2\a_1+3d_2^2\cos 2\a_1-3d_1^2\cos 2\a_2}{6 d_1 d_2 \cos(
		\a_1-\a_2)(\cos 2\a_1-\cos 2\a_2)\sqrt{\frac{-\cos2\a_1}{\cos2\a_2-\cos2\a_1}}\sqrt{\frac{\cos2\a_2}{\cos2\a_2-\cos2\a_1}}}, \label{iii.b:cost1}
	\\&&
	\cos \vartheta_2=\frac{(d_2^2-d_1^2)\cos 2\a_1 \cos 2\a_2}{2 d_1 d_2 \cos(
		\a_1+\a_2)(\cos 2\a_1-\cos 2\a_2)\sqrt{\frac{-\cos2\a_1}{\cos2\a_2-\cos2\a_1}}\sqrt{\frac{\cos2\a_2}{\cos2\a_2-\cos2\a_1}}}. \label{iii.b:cost2}
	\end{eqnarray}
	They imply that
	\begin{eqnarray}
	d_1^2=\frac{3 d_2^2 \sin 2 (\a_1-\a_2)+(2-3 d_2^2)\sin 2(\a_1+\a_2)}{6 \sin 2 \a_1 \cos 2\a_2}, \notag \\
	\end{eqnarray}
	so we know
	$(3 d_1^2-1)\sin 2\a_1 \cos 2\a_2+(3d_2^2-1)\cos 2\a_1 \sin 2\a_2=0$. Note that $\a_1\in (0,\p/4), \a_2\in(\p/4,\p/2)$ or $\a_1\in(\p/4,\p/2), \a_2\in (0,\p/4)$, $d_1,d_2,d_3$ are  nonnegative and real numbers. So $d_1=d_2=d_3=\frac{\sqrt{3}}{3}$.
	Then \eqref{iii.b:cost1}-\eqref{iii.b:cost2} imply that $\cos \vartheta_2=0$. Hence $\vartheta_2=\frac{\p}{2}$ or $\frac{3\p}{2}$.
	
Similarly, simplifying Eqs.\eqref{iii:mdl:h_41}-\eqref{iii:mdl:h_42} we have
	\begin{eqnarray}
	&&
	|e_1|^2\frac{\cos 2\a_2}{\cos 2\a_2-\cos 2\a_1}+|e_2|^2\frac{-\cos 2\a_1}{\cos 2\a_2-\cos 2\a_1}+\frac{\sqrt{-\cos2\a_1 \cos  2\a_2}}{|\cos2\a_2-\cos2\a_1|}[(-e^{i\vartheta_2} e_2)^\ast e_1+(-e^{i\vartheta_2}e_2)e_1^\ast]\cos(\a_1-\a_2)=\frac{1}{3}, \notag \label{iii.b:h12+h42}
	\\
	\\&&
	|e_1|^2\frac{\cos 2\a_1 \cos 2\a_2}{\cos 2\a_2-\cos 2\a_1}+|e_2|^2\frac{-\cos 2\a_1 \cos 2\a_2}{\cos 2\a_2-\cos 2\a_1}+\frac{\sqrt{-\cos2\a_1 \cos  2\a_2}}{|\cos2\a_2-\cos2\a_1|}[(-e^{i\vartheta_2} e_2)^\ast e_1+(-e^{i\vartheta_2}e_2)e_1^\ast]\cos(\a_1+\a_2)=0. \notag
	\label{iii.b:h12-h42}
	\\
	\end{eqnarray}
	\end{widetext}
	
	If $\vartheta_2=\frac{\p}{2}$ then $-e^{i \vartheta_2}=-i$. Then \eqref{iii.b:h12-h42} implies that
	\begin{eqnarray}
	&&
	\sqrt{-\cos 2\a_1 \cos 2\a_2}(|e_1|^2-|e_2|^2) \notag \\
	&&\pm (e_1 e_2^\ast-e_1^\ast e_2)\cos(\a_1+\a_2)i=0.
	\end{eqnarray}
	Since $\a_1\in (0,\p/4), \a_2\in(\p/4,\p/2)$ or $\a_1\in(\p/4,\p/2), \a_2\in (0,\p/4)$, we know $|e_1|=|e_2|,e_1 e_2^\ast-e_1^\ast e_2=0$. Applying them to \eqref{iii.b:h12+h42} we have $|e_1|=|e_2|=|e_3|=\frac{\sqrt{3}}{3}$. Then suppose $e_1=a+bi$. Since $e_1 e_2^\ast-e_1^\ast e_2=0$, we obtain $e_2=e_1=a+b i$. And \eqref{iii:di*ei} implies that $e_3=-2a-2bi$. It is a contradiction with  $|e_1|=|e_2|=|e_3|=\frac{\sqrt{3}}{3}$. Hence $W_3$ does not exist with  $\a_1\in (0,\p/4), \a_2\in(\p/4,\p/2)$ or $\a_1\in(\p/4,\p/2), \a_2\in (0,\p/4)$.
	
	If $\vartheta_2=\frac{3\p}{2}$, one can similarly show that $W_3$ does not exist with  $\a_1\in (0,\p/4), \a_2\in(\p/4,\p/2)$ or $\a_1\in(\p/4,\p/2), \a_2\in (0,\p/4)$.
\end{proof}
	
	\bibliographystyle{unsrt}
	
	\bibliography{20190319_sr3chm=one_zero_entry_in_V_and_W_}
	
\end{document}